\documentclass{article}

\usepackage{graphicx,amsmath,amsfonts,verbatim,amsthm,color,soul,xcolor,color,cleveref,authblk}

\graphicspath{{pictures/}}

\usepackage{tikz}
\usepackage{subcaption}
\usepackage{mathtools}
\usepackage[labelformat=simple]{subcaption}

\newtheorem{theorem}{Theorem}[section]
\newtheorem*{theorem*}{Theorem}
\newtheorem{proposition}[theorem]{Proposition}
\newtheorem{definition}[theorem]{Definition}
\newtheorem{lemma}[theorem]{Lemma}

\theoremstyle{definition}

\newtheorem{example}[theorem]{Example}
\newtheorem{question}[theorem]{Question}

\newcommand{\Z}{{\mathbb Z}}

\newcommand{\F}{{\mathbb F}}
\newcommand{\Proj}{{\mathbb P}}

\newcommand{\nathan}[1]{{ \sf $\heartsuit\heartsuit\heartsuit$ Nathan: [#1]}}

\begin{document}
\title{Representations of the Multicast Network Problem\thanks{The authors would like to acknowledge the hospitality of IPAM and the organizers of its Algebraic Geometry for Coding Theory and Cryptography Workshop: Everett Howe, Kristin Lauter, and Judy Walker.} \thanks{The third author was supported by NSA Young Investigator Grant H98230-16-10305 and by an AMS-Simons Travel Grant. The fourth author was partially supported by the CONACyT under grant title "Network Codes". The fifth uthor was partially supported by the National Science Foundation under grant DMS-1547399}}
\author[1]{Sarah E. Anderson}
\affil[1]{University of St. Thomas} 

\author[2]{Wael Halbawi} 
\affil[2]{California Institute of Technology}

\author[3]{Nathan Kaplan } 
\affil[3]{University of California, Irvine}

\author[4]{Hiram H. L\'opez}  
\author[4]{Felice Manganiello}
\affil[4]{Clemson University}

\author[5]{Emina Soljanin}
\affil[5]{Rutgers University}

\author[6]{Judy Walker} 
\affil[6]{University of Nebraska--Lincoln}

\date{\today}

\maketitle

\begin{abstract}

We approach the problem of linear network coding for multicast networks from different perspectives. We introduce the notion of the coding points of a network, which are edges of the network where messages combine and coding occurs. We give an integer linear program that leads to choices of paths through the network that minimize the number of coding points. We introduce the code graph of a network, a simplified directed graph that maintains the information essential to understanding the coding properties of the network.  One of the main problems in network coding is to understand when the capacity of a multicast network is achieved with linear network coding over a finite field of size $q$.  We explain how this problem can be interpreted in terms of rational points on certain algebraic varieties.
\end{abstract}

\section{Introduction}

A combinational  network, or simply network, is represented by 
a directed acyclic graph $G = (\mathcal{V},\mathcal{E},\mathcal{S},\mathcal{R},\mathcal{S},\F_q)$ where:
\begin{itemize}
\item $\mathcal{V}$ is the vertex set.
\item $\mathcal{E}$ is the set of unit-capacity directed edges.
\item $\mathcal{S}\subset \mathcal{V}$ is the source set, meaning the set of vertices with in-degree $0$.
\item  $\mathcal{R}\subset \mathcal{V}$ is the receiver set, meaning the set of vertices with out-degree $0$ and without loss of generality, we assume $\mathcal{S}\cap \mathcal{R}=\emptyset$. 
\item $\F_q$ is a finite field with $q$ elements and the alphabet used for communication.
\end{itemize}
  
In a network $G$, the sources originate messages and send them through the network via the edges of $G$, which represent error-free point-to-point communication channels. The communication requirements of a network $G$ are nonempty subsets $D_R\subseteq \mathcal{S}$ for $R \in \mathcal{R}$ that represent collections of sources from which $R$ must receive a message.

\begin{definition}
A multicast network is a directed acyclic graph $G= (\mathcal{V},\mathcal{E},\mathcal{S},\mathcal{R},\F_q)$ with communication requirements $D_R=\mathcal{S}$ for every receiver $R\in \mathcal{R}$, meaning that  every receiver demands the message sent by every source.
\end{definition}  

For a given  network $G$, a collection of requirements $\{D_R\mid R\in \mathcal{R}\}$ is achievable if there exist an alphabet $\F_q$ and a communication technique for which each receiver can
reconstruct the messages sent by the sources it demands. 

Examples of communication techniques are routing and network coding. In routing, vertices of a network forward a choice of their incoming messages. Network coding generalizes routing, and vertices can forward combinations of their incoming messages. Linear network coding is a special case of network coding where messages are vectors of $\F_q^n$ and vertices forward $\F_q$-linear combinations of their incoming messages.

\subsection{Achieving multicast network requirements}

We briefly explain how to achieve the requirements of a multicast network and describe some open problems in this area. For a more detailed explanation of the subject, we refer the interested reader to \cite{ks11c}.

Let $G$ be a multicast network with only one receiver $R$. Then its communication requirement is achieved by routing the messages if and only if $|\mathcal{S}|= \mathrm{mincut}(\mathcal{S},R)$. This is a consequence of the result known as the edge-connectivity version of Menger’s
Theorem \cite{co09}, which states that $\mathrm{mincut}(\mathcal{S},R)$ is equal to the maximum number of edge-disjoint paths between the source set $\mathcal{S}$ and the receiver $R$.

If the multicast network has multiple receivers, then a necessary condition for the communication requirements to be achieved is that 
\[|\mathcal{S}|= \min_{R\in \mathcal{R}}\mathrm{mincut}(\mathcal{S},R).\]    
This constant is called the capacity of a multicast network and corresponds to the maximum number of messages that every receiver might be able to reconstruct from the communication. It it evident that achieving the capacity of a multicast network is equivalent to achieving its communication requirements.  
  
Routing does not generally achieve capacity for multicast networks.  
The network in Figure
\ref{butterflynetwork} with the requirements that both receivers receive the messages from both sources, is the simplest
nontrivial example of a multicast network with multiple sources and
multiple receivers for which routing does not achieve capacity. This network is called the butterfly network and was first introduced in \cite{li03}. In order for the communication to achieve capacity, it is necessary for the vertex connected to both $S_1$ and $S_2$ to combine its incoming messages, meaning that it performs network coding.

\begin{figure}[h]
\centering
\scalebox{.29}{\includegraphics{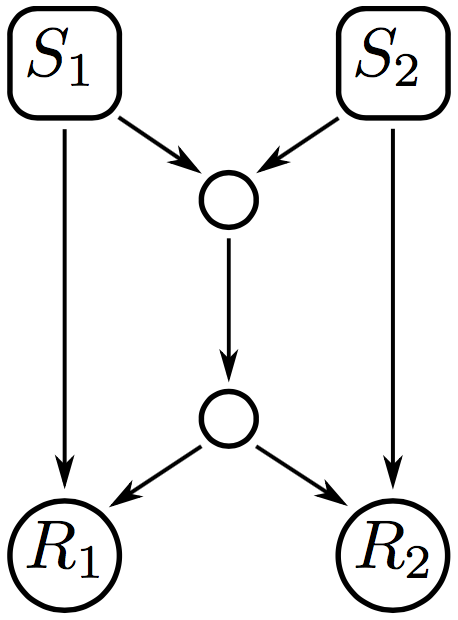}}
\caption{The butterfly network.}
\label{butterflynetwork}
\end{figure}

Li \emph{et al.} in \cite{li03} prove that the capacity of a multicast network is achieved by linear network coding. Let $N$ be the number of receivers of a given multicast network. In \cite{ko03}, K\"otter and M\'edard prove that linear network coding suffices to achieve capacity using vector spaces over any finite field $\F_q$ with size $q>N$. In \cite{ho06}, Ho \emph{et al.} showed that the capacity of a multicast network is achieved with high probability if the linear combinations are taken uniformly at random from a large finite field. Although every finite field of size $q>N$ is enough to achieve the capacity of every multicast network with $N$ receivers, it is believed that this bound is not tight. 

For a given multicast network, let $q_{\min}$ denote the smallest finite field size for which network capacity is achievable using linear network coding. It is important to note that achievability of capacity over $\F_q$ does not automatically imply achievability of capacity over larger fields.  Indeed, Sun \emph{et al.} showed that there exists a family of multicast networks such that $q_{\min}$ is bounded below by a constant times $\sqrt{N}$, but for which capacity is not achieveable over every field $\F_q$ with $q_{\min}<q\leq N$ \cite{su15}. In particular, they give a multicast network where capacity is achievable over $\F_7$ but not over $\F_8$.  We return to this example in Section \ref{s:vector_labeling}.  Therefore, in addition to determining $q_{\min}$, it is an important problem in multicast network communication to understand the set of $q$ with $q_{\min} < q \le N$ for which capacity is achievable using linear network coding over the finite field $\F_q$.

The aim of this work is to concisely provide different  representations of the multicast network problem for researchers to approach this problem from various viewpoints. The main background reference for this paper is the monograph \cite{fr07}.

For the rest of the work, we restrict to multicast networks for which capacity is achievable, meaning that the number of sources is equal to the capacity of the network. This condition is equivalent to the
existence for each receiver $R \in \mathcal{R}$ of a set of edge-disjoint paths $\mathcal{P}_R = \{P_{S,R} \mid S \in \mathcal{S} \}$, where $P_{S,R}$ is a path from $S$ to $R$.

The paper is organized as follows. In Section \ref{s:coding_pts}, we define coding points as the bottlenecks of the network where the linear combinations occur. We provide an integer optimization algorithm that, given a network, returns the subgraph of the network corresponding to a choice of paths between sources and receivers that uses the smallest possible number of coding points. 
Section \ref{s:code_graph} is dedicated to code graphs, which may be thought of as skeletons of multicast networks. More precisely, a code graph is a labeled directed graph whose vertices correspond to sources and coding points, with vertex labels representing the edge-disjoint paths from sources to receivers. We give properties for a labeled, directed acyclic graph to be the code graph of a multicast network. 
In the multicast network problem we label sources with vectors over $\F_q$ and specify how incoming messages combine at coding points. Capacity is achievable if and only if there is a labeling where these combinations satisfy certain linear independence conditions.  In Section \ref{s:vector_labeling}, we describe how these conditions translate to conditions on the vector labelings of the code graph of the network. This leads to the study of matrices over $\F_q$ with prescribed linear dependence and independence conditions, and then to the study of $\F_q$-rational points on certain varieties. We discuss the problem of determining the set of finite fields $\F_q$ for which such an $\F_q$-vector labeling of a code graph exists, and how to understand the collection of all such labelings as the solutions to systems of algebraic equations.

\section{Coding Points and Reduced Multicast Networks}\label{s:coding_pts}

Let $G=(\mathcal{V},\mathcal{E},\mathcal{S},\mathcal{R})$ be the
underlying directed acyclic graph of a multicast network.  In this section, we
study the coding points of $G$, which are the directed edges 
transmitting nontrivial linear combinations of the messages received at their tails. We are interested in choosing the set of
edge-disjoint paths that uses
the smallest possible number of coding points.

\begin{definition}\label{d:coding_point}
  Let $G$ be the underlying directed graph of a multicast network and for each $R\in \mathcal{R}$ let $\mathcal{P}_R = \{P_{S,R} \mid S \in \mathcal{S} \}$ be a set of edge-disjoint paths, where $P_{S,R}$ denotes a path from $S$ to $R$. 
  A \emph{coding point} of $G$ is an edge
  $E=(V,V^\prime)\in \mathcal{E}$ such that:
  \begin{itemize}
  \item there are distinct sources $S$, $S'$ and distinct receivers $R$, $R'$ such that
    $E$ appears in both $P_{S,R} \in \mathcal{P}_R$ and
    $P_{S',R'} \in \mathcal{P}_{R'}$, and 
  \item if $(W,V)\in P_{S,R}$ and $(W',V)\in P_{S',R'}$, then $W\neq W'$.
  \end{itemize}
\end{definition}

\begin{example}\label{coding_points_example}
Consider the  multicast network with four sources and three receivers shown in Figure \ref{f:coding_pt_network}. To make the figure more readable, we omit from the network the directed edges to $R_1$ from $S_1$ and $S_4$, to $R_2$ from $S_1$ and $S_2$, and to $R_3$ from $S_2$, $S_3$, and $S_4$.  For each receiver in this network, there is exactly one set of edge-disjoint paths from the sources to that receiver; these are shown in \Cref{f:coding_pt_network_paths_a}, \Cref{f:coding_pt_network_paths_b}, and \Cref{f:coding_pt_network_paths_c}. Using Definition  \ref{d:coding_point}, we note that the coding points of the network are precisely the edges depicted in Figure \ref{f:coding_pt_network_cp}.
\end{example}
    
\begin{figure}[htb]
\begin{minipage}{.32\textwidth}
\fbox{\begin{subfigure}[t]{\textwidth}
  \centering
 \scalebox{.29}{\includegraphics{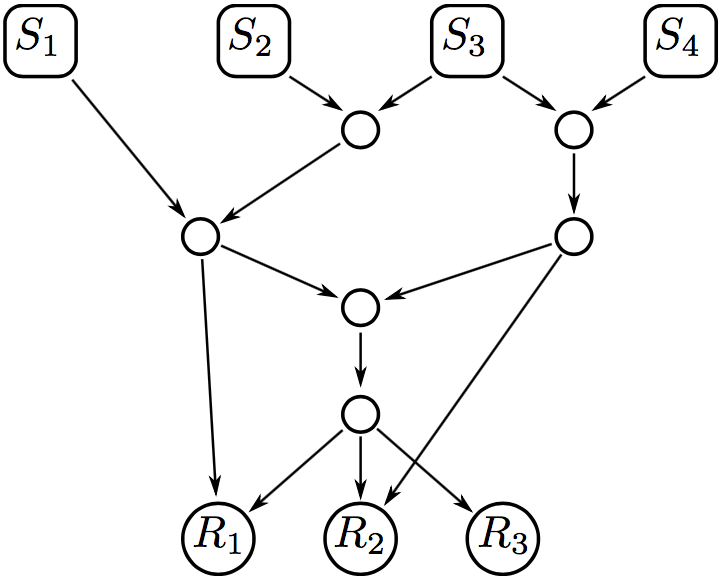}}
  \caption{A network with four sources and three receivers.  For readability, we have omitted edges $(S_1,R_1)$, $(S_4,R_1)$, $(S_1,R_2)$, $(S_2,R_2)$, $(S_2,R_3)$, $(S_3,R_3)$, and $(S_4,R_3)$.} 
  \label{f:coding_pt_network}
  \end{subfigure}}
\end{minipage}
\hfill
\begin{minipage}{.63\textwidth}
\begin{subfigure}[t]{0.48\textwidth}
 \centering
 \scalebox{.29}{\includegraphics{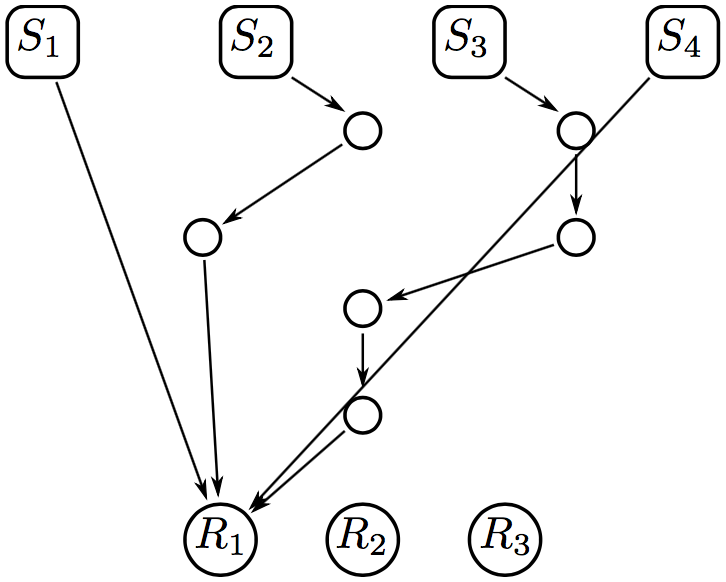}}
  \caption{$\mathcal{P}_{R_1}$.}
  \label{f:coding_pt_network_paths_a}
 \end{subfigure}\hfill
 \begin{subfigure}[t]{0.48\textwidth}
 \centering
 \scalebox{.29}{\includegraphics{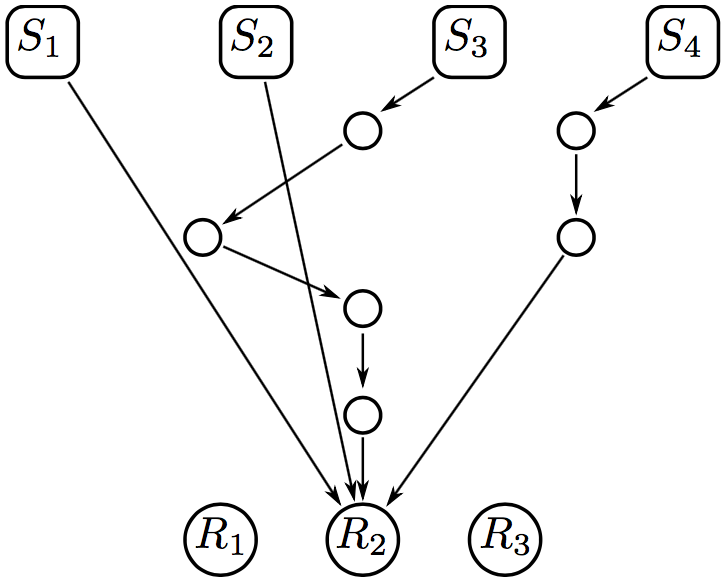}}
  \caption{$\mathcal{P}_{R_2}$.}
  \label{f:coding_pt_network_paths_b}
 \end{subfigure}
 
 \vspace{.2cm}
  \begin{subfigure}[t]{0.48\textwidth}
 \centering
 \scalebox{.29}{\includegraphics{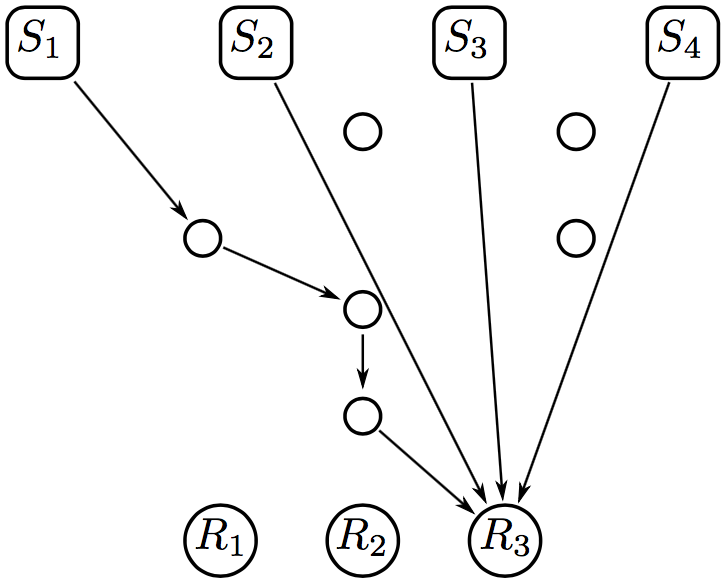}}
  \caption{$\mathcal{P}_{R_2}$.}
  \label{f:coding_pt_network_paths_c}
 \end{subfigure}\hfill
 \begin{subfigure}[t]{0.48\textwidth}
 \scalebox{.29}{\includegraphics{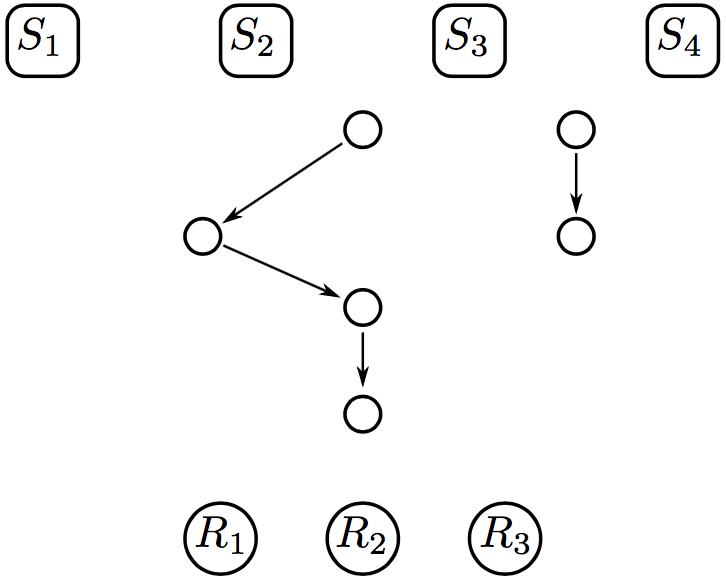}}
  \caption{Coding points.}
  \label{f:coding_pt_network_cp}
 \end{subfigure}
\end{minipage}
  \caption{A multicast network with four sources and three receivers.}
  \label{f:coding_pt_network_paths}
\end{figure}    
    
Network coding has been studied using linear programming. The Ford-Fulkerson algorithm from \cite{fo87} can be used to find the maximum flow of a unicast network, meaning a network with a source and a receiver. It can also be adapted to the case of a multicast network with only one receiver. In Section 3.5 of \cite{fr07}, the authors specify a linear optimization problem to maximize the flow of a network when using network coding.

Our goal is to specify an integer linear optimization problem whose solution corresponds to a set of paths $\mathcal{P}_R$, one for each $R\in \mathcal{R}$, with the minimal number of coding points.

  Without loss of generality, write
  $\mathcal{V}=\left\{V_1,\dots,V_N,\dots,V_t,V_{t+1},\dots,V_{t+L}\right\},$
  where $\mathcal{S}=\{V_1,\dots,V_N\}$ and
  $\mathcal{R}=\{V_{t+1},\dots,V_{t+L}\}$ for $N>0$ and $L>0.$
  Let $\mathrm{Adj}(G)$ be the adjacency matrix of the directed graph
  $G$. For every choice of $V_k\in \mathcal{S}$ and $V_{t+\ell}\in \mathcal{R}$,
  construct the $|\mathcal{V}| \times |\mathcal{V}|$ matrix
  \[X_{k\ell}:=\mathrm{Adj}(G,x_{k\ell})-\mathrm{Adj}(G,x_{k\ell})^T\]
  where  $\mathrm{Adj}(G,x_{k\ell})$ is obtained from
  $\mathrm{Adj}(G)$ by replacing each nonzero entry in position $(i,j)$ by the variable $x_{k\ell
    ij}$. Let ${\bf 1}$ be the vector of ones of size $|\mathcal{V}|$.

\begin{theorem}\label{t:optimization}
Let $G = (\mathcal{V}, \mathcal{E},\mathcal{S},\mathcal{R})$ be the
 underlying graph of a multicast network where $\mathcal{V}=\{V_i \mid i=1,\dots, t+L\}$, $\mathcal{S}=\{V_i\mid i=1,\dots,N\}$ and $\mathcal{R}=\{V_{t+i}\mid i=1,\dots,L\}$. 
Consider the following integer optimization problem:
\begin{eqnarray}
  \label{09.30.16.01} \mathrm{Minimize} && |\mathcal{E}| -
                                           \sum_{\substack{(V_i,V_j)\in\mathcal{E}}} z_{ij}\\
  \mathrm{Subject\  to}\nonumber && \forall \ k=1,\dots, N,\
                                    \ell=1,\dots,L, \  (V_i,V_j)\in
                                    \mathcal{E},\\
\label{09.30.16.02} X_{k\ell}\cdot{\bf 1}&=& 
[\underbrace{0,\dots,0,\overbracket{1}^{k\text{-th}},0,\dots,0}_N,\underbrace{0,\dots,0}_{t-N},
 \underbrace{0,\dots,0,\overbracket{-1}^{\ell\text{-th}},0,\cdots,0}_L]^T\\
\label{09.30.16.05} x_{k\ell ij}, z_{ij}&\in&\left\{0,1\right\}\\
\label{09.30.16.03}\sum_{k'=1}^N x_{k'\ell ij}&\leq& 1,\\
\label{09.30.16.04} z_{ij}&\leq&4-x_{k\ell i'i}-x_{k\ell ij}-x_{k^\prime\ell^\prime i''i}-x_{k^\prime\ell^\prime ij},\\
\nonumber && \qquad \qquad \forall \ \ k^\prime>k,\ \ell^\prime\neq\ell;  
\ (V_{i'},V_i)\neq(V_{i''},V_i)\in \mathcal{E}
\end{eqnarray}
where the $x$'s and the $z$'s are variables. 

A solution $x^*_{k\ell ij}, z^*_{ij}\in \{0,1\}$ of this optimization
problem corresponds to a collection of sets of edge-disjoint paths, one for each $\ell$,
\begin{eqnarray*}
\mathcal{P}_{V_{t+\ell}}=\{P_{V_kV_{t+\ell}} \, | \, k=1,\dots,N\}\label{eq:1}
  \label{e:paths}
\end{eqnarray*}
where
$P_{V_k V_{t+\ell}}=\{(V_i,V_j)\in \mathcal{E} \mid x^*_{k\ell ij}=1\}$ is
a path from $V_k$ to $V_{t+\ell}$. 

The number of coding points of
$G$ used by these collections of paths is 
\[ 
|\mathcal{E}| -\sum_{\substack{(V_i,V_j)\in\mathcal{E}}} z^*_{ij}.
\]
No other choice of paths gives a smaller number of coding points.

\end{theorem}
\begin{proof}
  Observe that variables $x_{k\ell ij}$ run over $V_k\in \mathcal{S}$,
  $V_{t+\ell}\in \mathcal{R}$ and $(i,j)$ where
  $(V_i,V_j)\in \mathcal{E}$. We begin by analyzing the meaning of each
  constraint of the optimization problem. Throughout the proof we say that an edge $(V_i,V_j)$ has value $1$ if the corresponding entry of the solution $x^*_{k\ell ij}$ is equal to $1$, for a given source $V_k$ and receiver $V_{t+\ell}$ and for a solution $x^*_{k\ell ij}\in \{0,1\}$ of the optimization problem.
  
\begin{itemize}
\item Constraint \eqref{09.30.16.02} consists of  $NL$ linear systems of equations, and its solutions correspond to paths between sources and receivers. \\
A given system of $\eqref{09.30.16.02}$ represents the constraints of the communication flow between source $V_k\in \mathcal{S}$ and receiver $V_{t+\ell}\in \mathcal{R}$. The system has $t+L$ equations. Each equation corresponds to the flow of a vertex of the network, meaning the sum of the values of its outgoing edges minus the sum of the values of its incoming edges. These equations divide into three types:
\begin{itemize}
\item The first $N$ equations correspond to the flow starting from the $N$ sources.  The only equation that is equal to $1$ is the $k$\textsuperscript{th} equation and the remaining are zero. This means that, given a solution of the system, the only edges with possibly nonzero value are the ones leaving source $V_k$. Since a solution has only entries in \{0,1\} and the $k$\textsuperscript{th} equation sums those entries to $1$, a solution has entry $1$ corresponding to only one of the edges leaving $V_k$ and $0$ otherwise.
\item The last $L$ equations correspond to the flow getting into the $L$ receivers. The only equation that is equal to $-1$ is the $\ell$\textsuperscript{th} equation and the remaining are zero. This means that, given a solution of the system, the only edges with possibly nonzero value are the ones leaving source $V_{t+\ell}$. Since a solution has only entries in \{0,1\} and the $\ell$\textsuperscript{th} equation subtracts those entries to $-1$, a solution has entry $1$ corresponding to one of the edges leaving $V_{t+\ell}$ and $0$ otherwise. 
T
\item The remaining equations are always equal to zero and correspond to the conservation of flow. This means that,  at each remaining vertex, the number of incoming edges with value $1$ equals the number of outgoing edges with value $1$.   
\end{itemize}
It follows that every solution of the system is a path $P_{V_k V_{t+\ell}}$ and that every path $P_{V_k V_{t+\ell}}$ satisfies the previous conditions.

\item Constraint \eqref{09.30.16.03} says that, given a receiver $V_{t+\ell}$ and an edge $(V_i,V_j)$, as we vary over all of the sources, that is, vary over all values of $k$, at most one of the entries of the solution $x^*_{k\ell ij}$ has value $1$. This inequality implies that the sets $\mathcal{P}_{V_{t+\ell}}$ with $\ell\in \{1,\dots,L\}$ obtained from the solution consists of edge-disjoint paths. 
\end{itemize}

Since $G$ is a multicast network by hypothesis, there exist  sets of edge-disjoint
paths $\mathcal{P}_{V_{t+\ell}}=\{P_{V_kV_{t+\ell}}\mid k=1,\dots,N\}$, one for each $V_{t+\ell}\in \mathcal{R}$. This implies that a solution of the linear optimization system satisfying \eqref{09.30.16.02}, \eqref{09.30.16.05} and
\eqref{09.30.16.03} exists.

Constraint \eqref{09.30.16.04} deals with the coding points. From Definition \ref{d:coding_point}, if $(V_i,V_j)\in \mathcal{E}$ is a coding point for a solution $x^*_{k\ell ij}\in \{0,1\}$,  then there exist 
distinct sources $V_k,V_{k^\prime}\in \mathcal{S},$  distinct receivers
$V_{t+\ell},V_{t+\ell^\prime}\in \mathcal{R}$, and  distinct edges
$(V_{i'},V_i),(V_{i''},V_i)\in \mathcal{E}$ such that $x^*_{k\ell i'j}=x^*_{k\ell ij}=
x^*_{k^\prime\ell^\prime i''j}=x^*_{k^\prime\ell^\prime ij}=1$. As a consequence $z^*_{ij}=0$. Instead, if $(V_i,V_j)\in \mathcal{E}$ is not a coding point for a solution $x^*_{k\ell ij}\in \{0,1\}$ then $z^*_{ij}\leq 1$. 
Condition \eqref{09.30.16.01} combined with constraints
\eqref{09.30.16.05} and \eqref{09.30.16.04}  implies that $z^*_{ij}=0$
if  $(V_i,V_j)\in \mathcal{E}$ is a coding point and $z^*_{ij}=1$ otherwise. Therefore, the number of coding points for a fixed assignment of the variables
$x^*_{k\ell ij}$ is 
\[|\mathcal{E}| -\sum_{\substack{(V_i,V_j)\in\mathcal{E}}} z^*_{ij}.\]

\end{proof}

We illustrate the algorithm with an example.

\begin{example}\label{butterfly_ex}
Consider the modified butterfly network of Figure \ref{modifiedbutterfly}. 
\begin{figure}[htb]
\centering
\begin{subfigure}[t]{0.48\textwidth}
\centering
\scalebox{.27}{\includegraphics{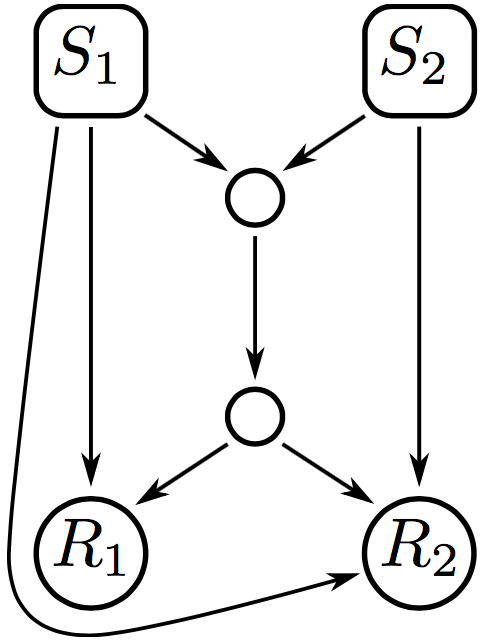}}
\caption{Modified butterfly network.}
\label{modifiedbutterfly}
\end{subfigure}
\begin{subfigure}[t]{0.48\textwidth}
\centering
\scalebox{.27}{\includegraphics{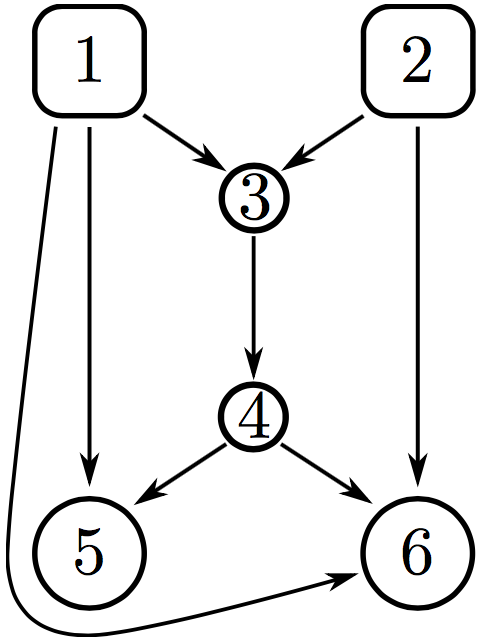}}
\caption{Labels for algorithm.}
\end{subfigure}

\vspace{.2cm}
\begin{subfigure}[t]{0.24\textwidth}
        \centering
        \scalebox{.29}{\includegraphics{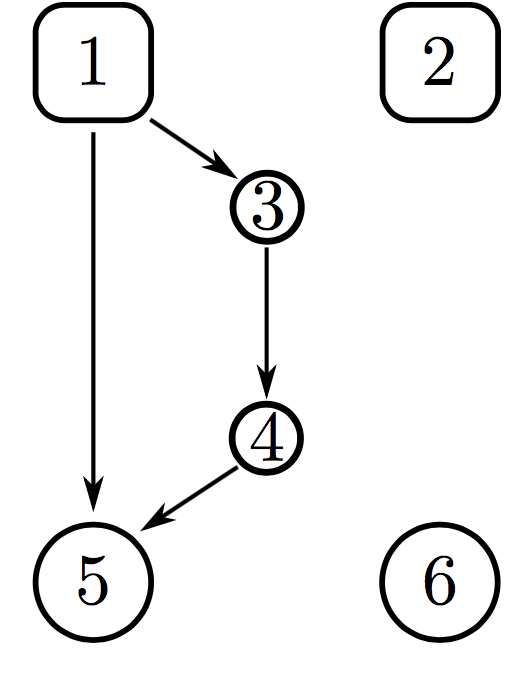}}
        \caption{Paths $P_{S_1,R_1}$}
        \label{modifiedbutterfly_a}
\end{subfigure}
\begin{subfigure}[t]{0.24\textwidth}
        \centering
        \scalebox{.29}{\includegraphics{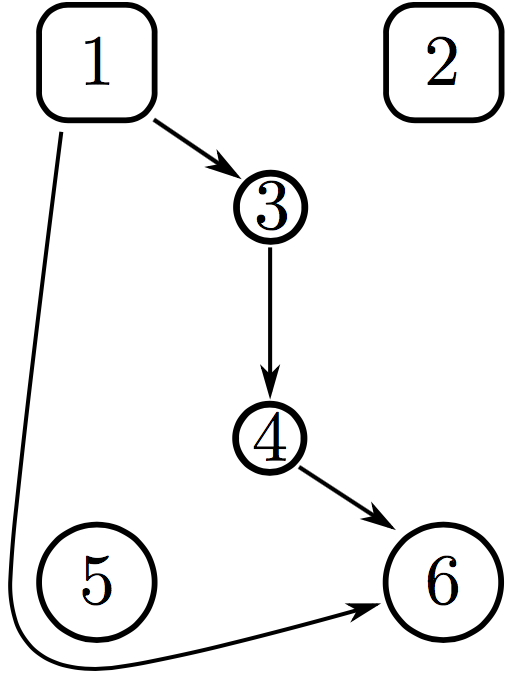}}
        \caption{Paths $P_{S_1,R_2}$}
        \label{modifiedbutterfly_b}
\end{subfigure}
\begin{subfigure}[t]{0.24\textwidth}
        \centering
        \scalebox{.29}{\includegraphics{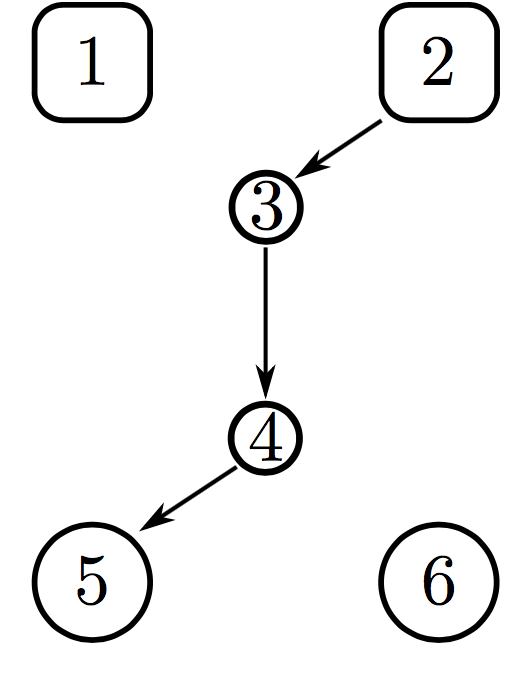}}
        \caption{Path $P_{S_2,R_1}$}
        \label{modifiedbutterfly_c}
\end{subfigure}
\begin{subfigure}[t]{0.24\textwidth}
        \centering
        \scalebox{.29}{\includegraphics{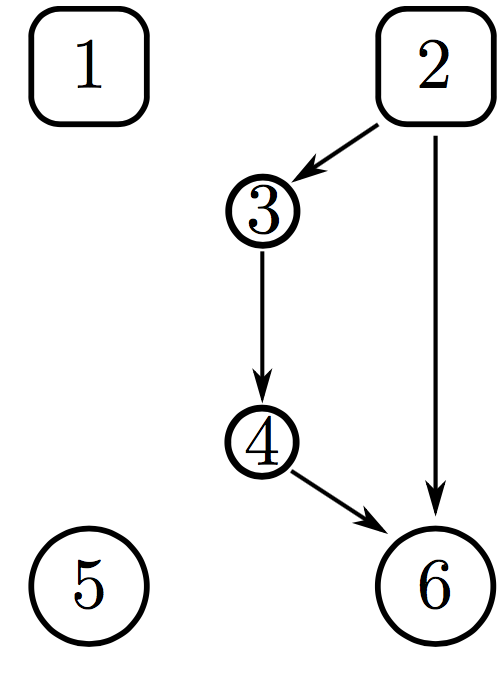}}
        \caption{Paths $P_{S_2,R_2}$}
        \label{modifiedbutterfly_d}
\end{subfigure}
\caption{Modified Butterfly Network}
\end{figure}
In this case, we have $N=2$ sources and $L=2$ receivers. As explained in the discussion leading up to the statement of Theorem \ref{t:optimization}, for each  $k,l\in\{1,2\}$ we define the matrix
\[X_{kl}=
\left(\begin{array}{cccccc}
0 & 0 & x_{kl13} & 0 & x_{kl15} & x_{kl16} \\
0 & 0 & x_{kl23} & 0 & 0 & x_{kl26} \\
-x_{kl13} & -x_{kl23} & 0 & x_{kl34} & 0 & 0 \\
0 & 0 & -x_{kl34} & 0 & x_{kl45} & x_{kl46} \\
-x_{kl15} & 0 & 0 & -x_{kl45} & 0 & 0 \\
-x_{kl16} & -x_{kl26} & 0 & -x_{kl46} & 0 & 0\end{array}\right).\]
Constraint \eqref{09.30.16.02} gives the following four linear systems of
equations:
\begin{eqnarray}
\label{11.30.1}X_{11}\cdot{\bf 1}&=&\left[1, 0, 0, 0, -1, 0\right]\\
\label{11.30.2}X_{12}\cdot{\bf 1}&=&\left[1, 0, 0, 0,0, -1\right]\\
\label{11.30.3}X_{21}\cdot{\bf 1}&=&\left[0, 1, 0, 0, -1, 0\right]\\
\label{11.30.4}X_{22}\cdot{\bf 1}&=&\left[0, 1, 0, 0, 0, -1\right].
\end{eqnarray}

Constraint \eqref{09.30.16.05} imposes the solutions of the system to be in $\{0,1\}$. It holds that:
\begin{itemize}
\item System~{\rm (\ref{11.30.1})} has two solutions (see Figure \ref{modifiedbutterfly_a}):
\begin{align*}
\mathbf{S}_{11}(1)&:=\left\{x_{1113}=x_{1134}=x_{1145}=1 \text{ and the rest } 0\right\} \text{ and} \\
\mathbf{S}_{11}(2)&:=\left\{x_{1115}=1 \text{ and the rest } 0\right\}.
\end{align*}
\item System~{\rm (\ref{11.30.2})} has two solutions (see Figure \ref{modifiedbutterfly_b}):
\begin{align*}
\mathbf{S}_{12}(1)&:=\left\{x_{1213}=x_{1234}=x_{1246}=1 \text{ and the rest } 0\right\} \text{ and }\\
\mathbf{S}_{12}(2)&:=\left\{x_{1216}=1 \text{ and the rest } 0 \right\}.
\end{align*}
\item System~{\rm (\ref{11.30.3})} has one solution (see Figure \ref{modifiedbutterfly_c}):
\begin{align*}
\mathbf{S}_{21}(1) &:=\left\{x_{2123}=x_{2134}=x_{2145}=1 \text{ and the rest } 0\right\}.
\end{align*}
\item System~{\rm (\ref{11.30.4})} has two solutions (see Figure \ref{modifiedbutterfly_d}):
\begin{align*}
\mathbf{S}_{22}(1)&:=\left\{x_{2223}=x_{2234}=x_{2246}=1 \text{ and the rest } 0\right\} \text{ and } \\
\mathbf{S}_{22}(2)&:=\left\{x_{2226}=1 \text{ and the rest } 0\right\}.
\end{align*}
\end{itemize}

Combining with constraint~\eqref{09.30.16.03}, we obtain three possible solutions:
\begin{eqnarray*}
\mathbf{S}_1&=&\mathbf{S}_{11}(2)\cup \mathbf{S}_{12}(1) \cup \mathbf{S}_{21}(1) \cup \mathbf{S}_{22}(2)\\
\mathbf{S}_2&=&\mathbf{S}_{11}(2)\cup \mathbf{S}_{12}(2) \cup \mathbf{S}_{21}(1) \cup \mathbf{S}_{22}(1)\\
\mathbf{S}_3&=&\mathbf{S}_{11}(2)\cup \mathbf{S}_{12}(2) \cup \mathbf{S}_{21}(1) \cup \mathbf{S}_{22}(2).
\end{eqnarray*}
Finally, for every edge $(i,j)$ of the network, we define the variable $z_{ij}.$ The expression in
\eqref{09.30.16.01}  tells us which of one the solutions $\mathbf{S}_1, \mathbf{S}_2$ or $\mathbf{S}_3$ has the minimum number of coding points.

If we take $\mathbf{S}_1,$ then constraint $\eqref{09.30.16.04}$ says that $z_{34}=0$
since $x_{1213}=x_{1234}=x_{2134}=x_{2123}=1.$ There is no restriction on the remaining $z_{ij}$'s for $S_1$, meaning that each one can be either $0$ or $1$. Then the minimum value of $|\mathcal{E}| - \sum_{\substack{(V_i,V_j)\in\mathcal{E}}} z_{ij}$
is $8-7=1$. This implies that $\mathbf{S}_1$ is a choice of paths with one coding point, $(3,4)\in \mathcal{E}$.

If we take either $\mathbf{S}_2$ or $\mathbf{S}_3,$ then constraint $\eqref{09.30.16.04}$ implies that there are no restrictions on the $z_{ij}$'s. Thus, the minimum value of
$|\mathcal{E}| - \sum_{\substack{(V_i,V_j)\in\mathcal{E}}} z_{ij}$
is $8-8=0.$ This implies that $\mathbf{S}_2$ and $\mathbf{S}_3$ are two choices of paths without coding points.  The algorithm in Theorem \ref{t:optimization} returns either $\mathbf{S}_2$ or $\mathbf{S}_3$.

\end{example}

 
\begin{definition}\label{reduced-network_definition} Let $G =
  (\mathcal{V},\mathcal{E},\mathcal{S},\mathcal{R},\{\mathcal{P}_R \,
  | \, R \in \mathcal{R}\})$ be a multicast network.  We say $G$ is
  \emph{reduced} if every edge and vertex of $G$ appears in some path $P_{S,R}$.

\end{definition}

Note that the multicast network of Figure \ref{f:coding_pt_network}, where $\mathcal{P}_{R_1}$, $\mathcal{P}_{R_2}$ and $\mathcal{P}_{R_3}$ are as shown in Figures \ref{f:coding_pt_network_paths_a}, \ref{f:coding_pt_network_paths_b}, and \ref{f:coding_pt_network_paths_c}, respectively, is reduced since every edge and vertex appears in some path. However, the modified butterfly network shown in Figure \ref{modifiedbutterfly} is not reduced for any choice of paths. For any of the collections of paths $\mathbf{S}_1, \mathbf{S}_2$, or $\mathbf{S}_3$ from Example \ref{butterfly_ex}, there is at least one edge of the modified butterfly network that does not appear. Therefore, the modified butterfly network is not reduced.

Every multicast network $G$ with a choice of a set of paths $\{\mathcal{P}_R \, | \, R\in \mathcal{R}\}$, contains a reduced multicast network $G'$ as a subgraph. It is obtained from $G$ by  omitting the unused edges and vertices. The multicast network $G$ and its reduced form $G'$ have the same properties with respect to linear network coding.  

For the rest of this paper, we assume that
\[G = (\mathcal{V},\mathcal{E},\mathcal{S},\mathcal{R},\{\mathcal{P}_R
\, | \, R \in \mathcal{R}\})\]
is a reduced multicast network and 
for which the set $\{\mathcal{P}_R \mid R \in \mathcal{R}\}$
corresponds to a choice of sets of edge-disjoint paths with a minimal number of coding points.

\section{Code Graphs}\label{s:code_graph}

In this section we introduce the code graph of a multicast network, which is a directed graph with labeled vertices that preserves the information essential to understanding the properties of the network with respect to linear coding.  We want the edges in this new graph to correspond to paths in the original network along which the message being passed is constant.  A maximal such path must originate either at a source or at the tail of a coding point.  This motivates the next definition.

\begin{definition}\label{coding-direct-path_definition} Let $G = (\mathcal{V}, \mathcal{E}, \mathcal{S},\mathcal{R},\{\mathcal{P}_R \mid R \in \mathcal{R}\})$ be a multicast network.    A \emph{coding-direct path} in $G$ from $V_1\in \mathcal{V}$ to $V_2 \in \mathcal{V}$ is a path in $G$ from $V_1$ to $V_2$ that does not pass through any coding point in $G$, except possibly a coding point with tail $V_1$.
\end{definition}

Given a multicast network, we now construct our desired vertex-labeled directed graph that preserves the essential coding properties of the network.  

\begin{lemma}\label{code-graph_lemma}
  Let $G = (\mathcal{V},\mathcal{E},\mathcal{S},\mathcal{R},\{\mathcal{P}_R \mid R \in \mathcal{R}\})$ be a multicast network and let $\mathcal{Q}$ be its set of
  coding points. Let $\Gamma = \Gamma(G)$ be the vertex-labeled
  directed graph constructed as follows:
\begin{itemize}
\item The vertex set of $\Gamma$ is $\mathcal{S} \cup \mathcal{Q}$, i.e., there
  is one vertex in $\Gamma$ for each source and for each coding
  point of $G$. Given a vertex $V$ of $\Gamma$ we call the corresponding
  source or coding point of $G$ the \emph{$G$-object} of $V$.
\item The edge set of $\Gamma$ is the set of all ordered pairs
  of vertices of $\Gamma$ such that there is a coding-direct
  path in $G$ between the corresponding $G$-objects.  
\item Each vertex $V$ of $\Gamma$ is labeled with a subset $L_V$ of the set $\mathcal{R}$.  A receiver $R \in \mathcal{R}$ is contained in $L_V$ if and only if there is a coding-direct path in $G$ from the $G$-object of $V$ to $R$.

\end{itemize}
The following properties hold:\begin{enumerate}
\item $\Gamma$ is an acyclic graph.
\item The in-degree of every vertex in $\Gamma$ is either $0$ or at least $2$.  Moreover, the $G$-object of a vertex $V$ of $\Gamma$ is a
  source in $G$ if and only if the in-degree of $V$ is $0$.
\item For each $R \in \mathcal{R}$ it holds that
\begin{itemize}
\item the cardinality of the set $\mathbf{V}_R = \{V\mid R\in L_V\}$ of vertices
in $\Gamma$ for which $R \in L_V$ is $|\mathcal{S}|$, and 
\item the set $\Pi_R = \{\pi_{S,R} \mid S \in \mathcal{S}\}$, where $\pi_{S,R}$ is a path in $\Gamma$ from $S$ to $V_{S,R}$ corresponding to the path $P_{S,R}$ in $G$ from $S$ to $R$, consists of $|\mathcal{S}|$ vertex-disjoint paths. 
\end{itemize}
\end{enumerate}
\end{lemma}

\begin{proof} {\ } \begin{enumerate}
\item Suppose the vertices $V_1,\dots, V_n$ form a cycle in $\Gamma$.  Then the coding-direct paths in $G$ joining the
  corresponding $G$-objects form a cycle in $G$, which is impossible
  since $G$ is acyclic.
\item Let $V$ be a vertex of $\Gamma$.  The $G$-object
  of $V$ in $G$ is either a source or a coding point of
  $G$.  If it is a source $S \in \mathcal{S}$, then the in-degree of
  $S$ in $G$ is 0.  Thus $S$ cannot be the end of a coding-direct
  path in $G$, and so $V$ cannot be the head of an edge in
  $\Gamma$.  Thus the in-degree of $V$ is 0.  
  
  On the other hand, if the
  $G$-object of $V$ is a coding point $Q\in \mathcal{Q}$, then there are distinct sources $S, S' \in \mathcal{S}$ and
  distinct receivers $R, R' \in \mathcal{R}$ such that $Q$ appears in both paths
  $P_{SR}$ and $P_{S'R'}$ in $G$. If $Q$ is the first coding point appearing
  in the path $P_{SR}$, then there is a coding-direct path in $G$ from
  $S$ to the tail of $Q$. Hence there is an edge in $\Gamma$ from $S$ to $V$.
  Otherwise, there is some other coding point along the path $P_{S,R}$ such
  that there is a coding-direct path in $G$ from that coding point to $V$, and so there is
  an edge in $\Gamma$ from that coding point to $V$.  In either case, there is an incoming edge to $V$ from the path $P_{S,R}$, and the same holds for the path $P_{S',R}$.  Thus the in-degree of $V$ is at least 2, as desired. 
\item Fix $R \in \mathcal{R}$ and let $S \in \mathcal{S}$.  If $P_{S,R}$ is a coding-direct path, then $R \in L_S$. 
Otherwise, let $Q_{S,1}$,
\dots, $Q_{S,t_S}$ be the coding points in $P_{S,R}$ in order of appearance.  This implies $R\in L_{Q_{S,t_S}}$. The set $\mathbf{V}_R$ is the set of vertices of $\Gamma$ whose $G$-objects have a coding-direct path to $R$.

If $P_{S,R}$ is a coding-direct path, then $\pi_{S,R}$ is the empty path consisting of the single vertex $S$ in $\Gamma$. Otherwise, there is a path  $\pi_{S,R}$ in $\Gamma$ from $S$ to $Q_{S,t_S}$ since in $G$ there is a coding-direct path from  the $G$-object
of $S$ to the $G$-object of $Q_{S,1}$ and
from the $G$-object of $Q_{S,\ell}$ to the one of $Q_{S,\ell+1}$ for $1 \leq \ell \leq t_S-1$. Let $\Pi_R$ be the collection of these paths.\\
  Since the paths in $\mathcal{P}_R$ are
edge-disjoint, no coding point can appear in more than one  path $P_{S,R}$.  Thus the corresponding paths in $\Pi_R$ are vertex-disjoint. As a consequence $|\mathbf{V}_R|=|\Pi_R|=|S|$. 
\end{enumerate}
\end{proof}

\begin{definition} Let $G = (\mathcal{V},\mathcal{E},\mathcal{S},\mathcal{R},\{\mathcal{P}_R \mid R \in \mathcal{R}\})$ be a multicast network.  The graph $\Gamma = \Gamma(G)$ constructed in Lemma~\ref{code-graph_lemma} is called the \emph{code graph} of $G$.  To refer to a code graph and its associated data succinctly, we often say that $\Gamma = (\mathcal{S},\mathcal{Q},\mathcal{R},\{\mathbf{V}_R \mid R \in \mathcal{R}\},\{\Pi_R \mid R \in \mathcal{R}\})$ is the code graph of $G$. 
\end{definition}

Figure \ref{5a} shows the butterfly network; its code graph is shown in Figure \ref{5b}. Note that the code graph has three vertices since the butterfly network has two sources $S_1$ and $S_2$, and one coding point, which we call $Q$. There is an edge between the vertex corresponding to $S_1$ and the vertex corresponding to $Q$ since there is a coding-direct path in the butterfly network between $S_1$ and $Q$. There is a similar edge coming from the vertex corresponding to $S_2$. Finally, the vertex corresponding to $S_1$ is labeled with the receiver $R_1$ since there is a coding-direct path in the butterfly network from $S_1$ to $R_1$. The vertices corresponding to $S_2$ and to $Q$ are labeled similarly. 

Figure \ref{5c} shows the multicast network of Figure \ref{f:coding_pt_network} (recall that there are edges of the network not shown in this figure), where $\mathcal{P}_{R_1}$, $\mathcal{P}_{R_2}$ and $\mathcal{P}_{R_3}$ are as shown in Figures \ref{f:coding_pt_network_paths_a}, \ref{f:coding_pt_network_paths_b}, and \ref{f:coding_pt_network_paths_c}. Its code graph is shown in Figure \ref{5d}.

\begin{figure}[h]
\centering
\begin{subfigure}{.45\linewidth}
\centering
\scalebox{.29}{\includegraphics{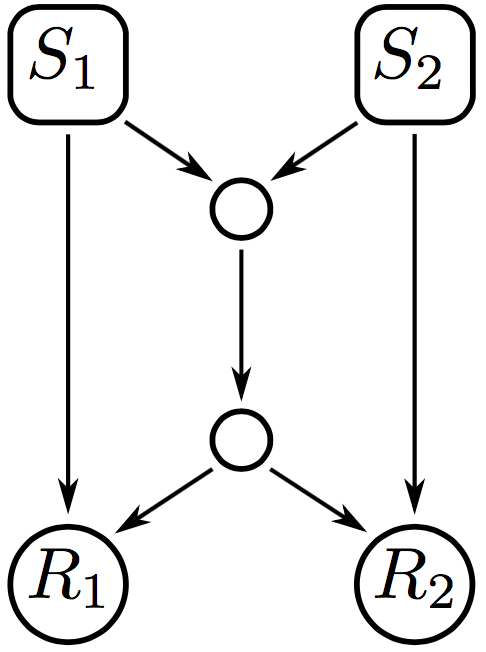}}
\subcaption{Butterfly network}
\label{5a}
\end{subfigure}
\begin{subfigure}{.45\linewidth}
\centering
\scalebox{.29}{\includegraphics{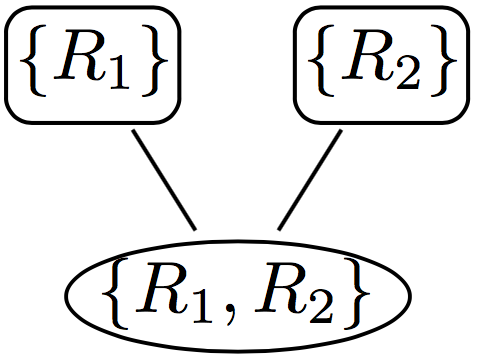}}
\subcaption{Code graph}
\label{5b}
\end{subfigure}

\vspace{.2cm}
\begin{subfigure}{.45\linewidth}
\centering
\scalebox{.29}{\includegraphics{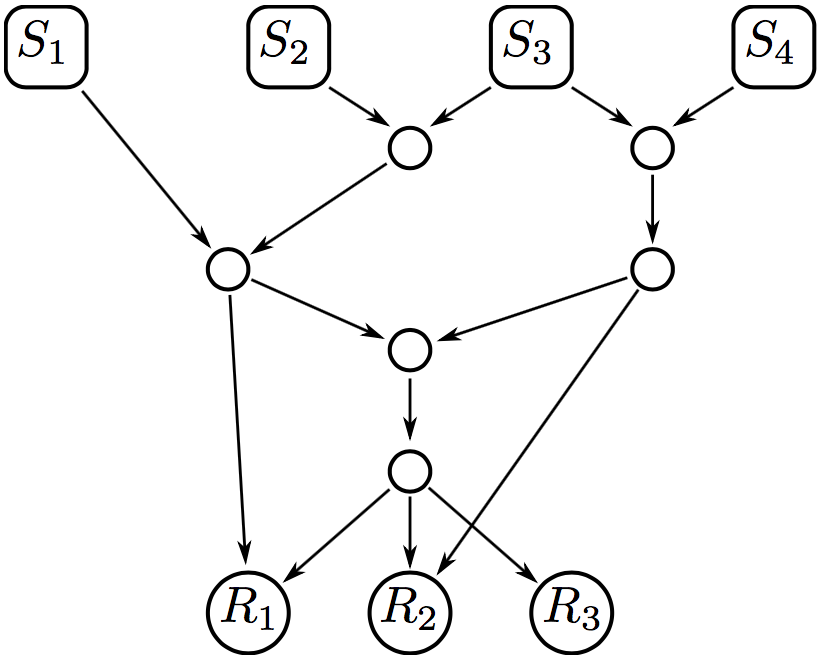}}
\subcaption{Multicast network}
\label{5c}
\end{subfigure}
\begin{subfigure}{.45\linewidth}
\centering
\scalebox{.29}{\includegraphics{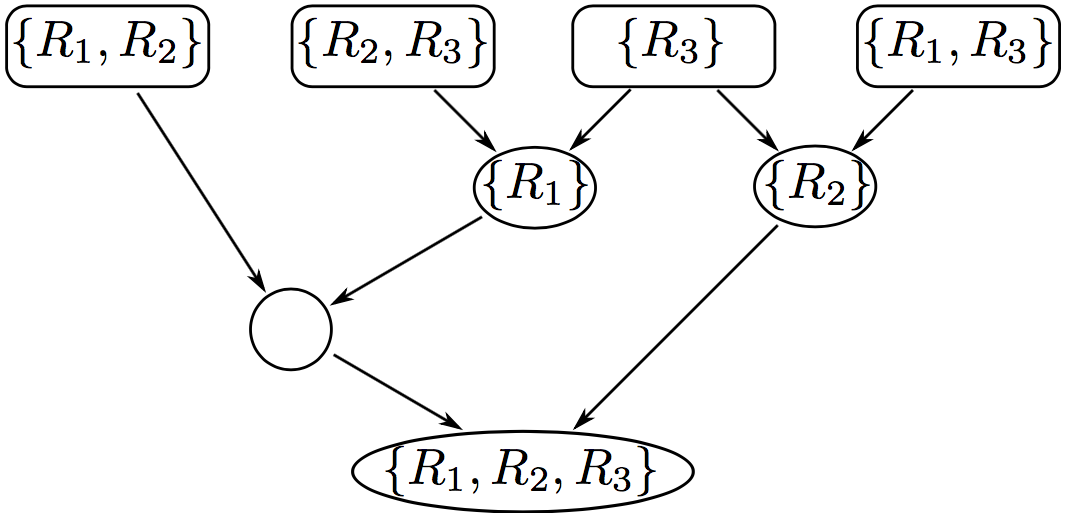}}
\subcaption{Code graph}
\label{5d}
\end{subfigure}
\caption{Networks with their code graph}
\label{butterflynetworkwithcodegraph}
\end{figure}


The following proposition characterizes the directed acyclic graphs that occur as the code graph of a multicast network.

\begin{proposition}\label{code-graph_is_network_proposition}
Let $\Gamma = (\mathbf{V}, \mathbf{E})$ be a vertex-labeled, directed acyclic graph where each vertex $V$ is labeled with a finite set $L_V$. Let $\mathcal{S}: = \{V \in \mathbf{V} \mid  V \text{ has in-degree } 0 \}$, $\mathcal{Q} := \mathbf{V} \backslash \mathcal{S}$, and  $\mathcal{R}$ the union of the sets $L_V$.
Suppose
\begin{itemize}
\item $\mathcal{S}$ and $\mathcal{R}$ are non-empty; 
\item the in-degree of every vertex in $\mathcal{Q}$ is at least 2;

\item each $R \in \mathcal{R}$ appears as a label of exactly $|\mathcal{S}|$ distinct vertices; and
\item for each $R \in \mathcal{R}$, there is a set $\Pi_R = \{\pi_{S,R} \, | \, S \in \mathcal{S}\}$ of vertex-disjoint paths such that $\pi_{S,R}$ starts at $S$ and ends at a vertex labeled with $R$, and every vertex and every edge of $\Gamma$ is contained $\pi_{S,R}$ for at least one pair $(S,R)$.

\end{itemize}

\noindent Then $\Gamma$ is the code graph for a reduced multicast network whose sources, coding points, and receivers are in one-to-one correspondence with the elements of $\mathcal{S}$, $\mathcal{Q}$, and $\mathcal{R}$, respectively.
\end{proposition}

\begin{proof}
Construct a
directed graph $G$ as follows:
\begin{itemize}
\item Create one vertex
  of $G$ for every $S \in \mathcal{S}$ and for every $R \in
  \mathcal{R}$, and create two vertices in $G$, $Q^{\text{tail}}$
  and $Q^{\text{head}}$, for every $Q \in
  \mathcal{Q}$. In other words, writing $\mathcal{Q}^{\text{tail}} = \{Q^{\text{tail}} \, | \, Q \in \mathcal{Q}\}$ and $\mathcal{Q}^{\text{head}} = \{Q^{\text{head}} \, | \, Q \in \mathcal{Q}\}$, the vertex set of $G$ is $\mathcal{V} = \mathcal{S} \cup \mathcal{R} \cup \mathcal{Q}^{\text{tail}}
  \cup \mathcal{Q}^{\text{head}}$. 
\item The edge set $\mathcal{E}$ of $G$ consists of five types of edges:
\begin{enumerate}
\item[(a)] For each $Q \in \mathcal{Q}$, create an edge in $G$ from
  $Q^{\text{tail}}$ to $Q^{\text{head}}$.
\item[(b)] For each edge in $\Gamma$ from $S \in \mathcal{S}$ to
  $Q \in \mathcal{Q}$, create an edge in $G$ from $S$ to $Q^{\text{tail}}$.
\item[(c)] For each edge in $\Gamma$ from $Q_1 \in \mathcal{Q}$ to
  $Q_2 \in \mathcal{Q}$, create an edge in $G$ from
  $Q_1^{\text{head}}$ to $Q_2^{\text{tail}}$.
\item[(d)] Let $R \in \mathcal{R}$ and $S \in \mathcal{S}$.  If $R \in L_S$, create an edge in $G$ from $S$ to $R$.
\item[(e)] Let $R \in \mathcal{R}$ and $Q \in \mathcal{Q}$. If $R \in L_Q$, create an edge in $G$ from $Q^{\text{head}}$ to $R$.
\end{enumerate}
\end{itemize}
Then for each $R \in \mathcal{R}$, the set $\Pi_R$ of vertex-disjoint paths in $\Gamma$ induces a set $\mathcal{P}_R$ of edge-disjoint paths in $G$ from the elements of $\mathcal{S}$ to $R$, and $G = (\mathcal{V},\mathcal{E},\mathcal{S},\mathcal{R},\{\mathcal{P}_R \, | \, R \in \mathcal{R}\})$ is a multicast network. Moreover,
the code graph of $G$ is $\Gamma$.
 We must show:
\begin{enumerate}
\item $\mathcal{S}$ is precisely the set of vertices of $G$ of in-degree
  0;
\item $\mathcal{R}$ is precisely the set of vertices of $G$ of
  out-degree 0;
\item $G$ is a directed acyclic graph;
\item For each $R \in \mathcal{R}$, $\Pi_R$ induces a set of edge-disjoint
  paths $\mathcal{P}_R = \{P_{SR} \, | \, S \in \mathcal{S}\}$, where
  $P_{SR}$ is a path in $G$ from $S$ to $R$; and
\item The code graph of $G = (\mathcal{V},\mathcal{E},\mathcal{S},\mathcal{R},\{\mathcal{P}_R \, | \, R \in \mathcal{R}\})$ is $\Gamma$.
\end{enumerate}
We do each of these in turn.

\begin{enumerate}
\item By construction, every $S \in \mathcal{S}$ has in-degree 0 in
  $G$, and every $R \in \mathcal{R}$ and every $Q^{\text{head}} \in
  \mathcal{Q}^{\text{head}}$ has strictly positive in-degree in $G$.
  Let $Q^{\text{tail}} \in \mathcal{Q}^{\text{tail}}$ come from the
  vertex $Q$ of $\Gamma$.  If $Q^{\text{tail}}$ has in-degree 0, then
  $Q$ has in-degree 0, and so $Q \in \mathcal{S}$, a contradiction.
  Thus, $\mathcal{S}$ is precisely the set of vertices of $G$ of
  in-degree 0.

\item By construction, every $R \in \mathcal{R}$ has out-degree 0 in
  $G$.  For every $S \in \mathcal{S}$, there is a path in $\Gamma$ from
  $S$ to some vertex of $\Gamma$ labeled with $R$ for each $R \in
  \mathcal{R}$; in particular this means that the out-degree of $S$ in
  $G$ cannot be 0.  By construction, no $Q^{\text{tail}} \in
  \mathcal{Q}^{\text{tail}}$ can have out-degree 0 in $G$. Let
  $Q^{\text{head}} \in \mathcal{Q}^{\text{head}}$ come from the vertex
  $Q$ of $\Gamma$. Since every vertex and every edge of $\Gamma$ is contained $\pi_{S,R}$ for at least one pair $(S,R)$, there exist some $R \in \mathcal{R}$ such that either there is an edge in $\Gamma$ from $Q$ to some $Q'$ or $Q$ is labeled in $\Gamma$ with $R$. In
  either case, the out-degree of $Q^{\text{head}}$ in $G$ is strictly
  positive.  Hence, $\mathcal{R}$ is precisely the set of vertices of
  $G$ of out-degree 0.

\item Suppose there is a cycle in $G$.  Since every $S \in
  \mathcal{S}$ has in-degree 0 in $G$ and every $R \in \mathcal{R}$
  has out-degree 0 in $G$, this cycle can include only vertices in
  $\mathcal{Q}^{\text{tail}} \cup \mathcal{Q}^{\text{head}}$ and hence
  only edges of types (a) and (c).  This means that, writing the
  cycle in terms of its vertices, it must have the form
\[
Q_1^{\text{tail}}, Q_1^{\text{head}}, Q_2^{\text{tail}},
Q_2^{\text{head}}, \dots,  Q_n^{\text{tail}}, Q_n^{\text{head}},
Q_1^{\text{tail}}.
\]
This yields a cycle $Q_1, Q_2, \dots, Q_n, Q_1$ in
$\Gamma$, contradicting the fact that $\Gamma$ is acyclic.

\item Let $R \in \mathcal{R}$ and let $\Pi_R = \{\pi_{S,R} \, | \, S \in \mathcal{S}\}$ be the set of vertex-disjoint paths in $\Gamma$ from the sources to $R$.  Since every vertex in $\mathcal{S}$ has in-degree 0 and
every vertex in $\mathcal{R}$ has out-degree 0, for each $S \in \mathcal{S}$ there are vertices
$Q_{S,1}, \dots, Q_{S,t_S} \in \mathcal{Q}$ such that $Q_{S,t_S} =
V_S$ and that path $\pi_{S,R}$ in $\Gamma$ consists of the following vertices in order:
\[
S, Q_{S,1}, \dots, Q_{S,t_S-1}, Q_{S,t_S}.
\]
This gives the path $P_{S,R}$ in $G$ consisting of
\[
S, Q_{S,1}^{\text{tail}}, Q_{S,1}^{\text{head}},  \dots,
Q_{S,t_S-1}^{\text{tail}}, Q_{S,t_S-1}^{\text{head}},
Q_{S,t_S}^\text{tail}, Q_{S,t_S}^\text{head}, R.
\]
Since for $S \neq S'$ the paths $\pi_{S,R}$ and $\pi_{S',R}$ in $\Gamma$ are vertex-disjoint, we have that $P_{S,R}$ and $P_{S',R}$ are edge-disjoint paths in
$G$.  Thus 
\[
\mathcal{P}_R = \{P_{S,R} \, | \, S \in \mathcal{S}\}
\]
is the set of edge-disjoint paths in $G$ we seek.

\item It is clear from Lemma~\ref{code-graph_lemma} that
  $\Gamma$ is the code graph of $G$.
  
\end{enumerate}
\end{proof}
Note that the proof of Proposition \ref{code-graph_is_network_proposition} provides a method to construct a multicast network with a given code graph. 

\section{$\F_q$-vector labelings and matrices}\label{s:vector_labeling}

We begin this section by recalling the goal of linear network coding on a multicast network.  We have a set $\mathcal{S}$ of sources that transmit their messages along the edges of a network to a set $\mathcal{R}$ of receivers.  Certain edges, those corresponding to coding points, can be used by multiple paths between sources and receivers.  These edges combine inputs in such a way that the output can be decoded.

Messages sent by sources are represented by elements of $\F_q^{n}$. When incoming messages combine along a coding edge the output is a linear combination of these input vectors.  Each receiver must obtain the message sent by each source, which motivates  linear independence constraints on these vectors.

\begin{definition} 
Let $G$ be a multicast network with source set $\mathcal{S}$ and receiver set $\mathcal{R}$, and let $\Gamma=(\mathbf{V},\mathbf{E})$ be its corresponding code graph.  Each $V \in \mathbf{V}$ is labeled with a subset $L_V \subseteq \mathcal{R}$.  Let $\mathbf{V}_R$ denote the set of $V\in \mathbf{V}$ such that $R \in L_V$.

An \emph{$\mathbb{F}_q$-vector labeling} of $\Gamma$ is an assignment of elements of $\mathbb{F}_q^{|\mathcal{S}|}$ to the vertices of $\Gamma$ satisfying:
\begin{itemize}
\item The vectors assigned to the source nodes of the code graph are linearly independent.
\item For each $R \in \mathcal{R}$ the vectors assigned to the vertices of $\mathbf{V}_R$ are linearly independent.
\item The vector assigned to a coding point $Q \in \mathbf{V}$ is in the span of the vectors assigned to the tails of the directed edges terminating at $Q$.
\end{itemize}
\end{definition}

The discussion of the linear network coding problem in the introduction combined with the results of Section \ref{s:coding_pts} implies the following.

\begin{proposition}
Let $G$ be a multicast network and let $\Gamma$ be its corresponding code graph.  The capacity of $G$ is achievable over $\F_q$ if and only if there exists an $\F_q$-vector labeling of $\Gamma$.
\end{proposition}

\begin{question}
Let $G$ be a multicast network and let $\Gamma$ be its corresponding code graph.
\begin{enumerate}
\item What is the smallest $q$ such that an $\F_q$-vector labeling of $\Gamma$ exists?
\item Can we describe the set of $\F_q$-vector labelings of $\Gamma$ algebraically?
\end{enumerate}
\end{question}

The first question has been studied extensively in the literature; see for example \cite{li03,ko03,ho06,fr07}.  Previous work on this problem has taken the approach of labeling the messages for the source nodes of the multicast network directly.  When messages pass through a common edge of the network, the edge is assigned a transfer matrix, which takes a linear combination of the incoming vectors.  The condition that each receiver obtains the message from each source corresponds to the transfer matrices being invertible.  See \cite{ho06} for a characterization of when the capacity of a multicast network is achievable over $\F_q$ in terms of linear combinations of determinants of these matrices.  The definition of $\F_q$-vector labeling allows us to bypass transfer matrices and work directly with linear dependence and independence conditions of a single matrix.

We illustrate this idea with several examples.  For a collection of vectors $v_1,\ldots, v_k \in \F_q^n$, we write $\langle v_1,\ldots, v_k\rangle$ for their span. Let $\mathbf{M}_{m\times n}(q)$ denote the set of $m \times n$ matrices with entries in $\F_q$.
\begin{example}\label{butterfly_label}
Consider the butterfly network and its code graph $\Gamma$, which are given in Figure \ref{butterflynetworkwithcodegraph}.  There are two source nodes $S_1$ and $S_2$, and one coding node $Q$. We label each of these three vertices with a vector in $\F_q^2$. Call these $v_{S_1}, v_{S_2}$, and $v_Q$.  Then $v_{S_1}$ and $v_{S_2}$ must be linearly independent, and $v_Q \in \langle v_{S_1}, v_{S_2}\rangle$.  Since $S_1$ and $Q$ share the receiver label $R_1$, we have that $v_{S_1}$ and $v_Q$ are linearly independent.  Similarly, $S_2$ and $Q$ share the receiver label $R_2$, so $v_{S_2}$ and $v_Q$ are linearly independent.  

We see that for a fixed $q$, the set of $\F_q$-vector labelings of $\Gamma$ are in bijection with those elements of $\mathbf{M}_{2\times3}(q)$ satisfying that each of the three $2\times 2$ minors are nonzero.  This condition is equivalent to saying that no two columns are scalar multiples of each other.  There are precisely $(q-1)^4(q+1)q$ ways to choose such a matrix.  For $q=2$ these six ways correspond to the six ways to permute the columns of the matrix $\left(\begin{smallmatrix} 1 & 0 & 1 \\ 0 & 1 & 1 \end{smallmatrix}\right)$. 
\end{example}

In order to explain the counting formula from this example, we recall the definition of projective space over a finite field.  
\begin{definition}
Let $n \ge 1$ and $\F_q$ be a finite field.  The points $\Proj^n(\F_q)$ correspond to elements of $\F_q^{n+1} \setminus (0,\ldots, 0)$ up to the equivalence that $(a_0,\ldots, a_n) \sim (\alpha a_0,\ldots, \alpha a_n)$ for every $\alpha \in \F_q^*$.  
\end{definition}
We see that $\left| \Proj^n(\F_q) \right| = \frac{q^{n+1}-1}{q-1}$.  Two nonzero columns of a $2\times n$ matrix are linearly dependent if and only if they define the same point in $\Proj^1(\F_q)$.  So, in Example \ref{butterfly_label} we need only count the number of ways to pick three distinct points in $\Proj^1(\F_q)$, and then multiply by $(q-1)^3$ to account for possible scalings of each column.

\begin{example}\label{vector_label_ex2}
Let $\Gamma$ be the code graph from Figure \ref{5d}. There are $8$ vertices of $\Gamma$, four of which correspond to sources. Number these vertices so that \#1-4 are the nodes in the top row, in order; \#5-7 are the nodes in the middle row, in order; and \#8 is the bottom-most node.  An $\F_q$-vector labeling gives an element of $\mathbf{M}_{4\times 8}(q)$ with columns $v_1, \ldots, v_8$ satisfying the following conditions:
\begin{enumerate}
\item The labels of the source nodes are linearly independent, and the labels of every set of nodes labeled with a common receiver are linearly independent.  Therefore, the submatrices with the following sets of columns are invertible:
\[
\{v_1,v_2,v_3,v_4\}, \{v_1,v_4,v_6,v_8\}, \{v_1,v_2,v_7,v_8\}, \{v_2,v_3,v_4,v_8\}.
\]
\item The label of every vertex is in the span of the labels of the vertices that are tails of those directed edges terminating at that vertex, which implies:
\[
v_5 \in \langle v_1, v_6\rangle,\ v_6 \in \langle v_2, v_3\rangle,\ v_7 \in \langle v_3, v_4\rangle,\ v_8 \in \langle v_5, v_7\rangle.
\]
\end{enumerate}

Without loss of generality, we may label the source nodes so that the matrix with columns $v_1, v_2, v_3, v_4$ is the identity matrix.  This leads to a $4\times 8$ matrix of the form
\[
\left( 
\begin{array}{cccccccc}
1&0&0&0&x_5&0&0&x_8 \\
0&1&0&0&\alpha y_6&y_6&0&y_8 \\
0&0&1&0&\alpha z_6&z_6&z_7&z_8 \\
0&0&0&1&0&0&w_7&w_8
\end{array}
\right),
\]
subject to the additional constraint that $v_8 \in \langle v_5, v_7\rangle$, and that each of the matrices
\begin{eqnarray*}
M_{1468} & = &  \left(
\begin{array}{cccc}
1&0&0&x_8\\
0&0&y_6&y_8\\
0&0&z_6&z_8\\
0&1&0&w_8
\end{array}
\right),\\
M_{1278} & = &  \left(
\begin{array}{cccc}
1&0&0&x_8\\
0&1&0&y_8\\
0&0&z_7&z_8\\
0&0&w_7&w_8
\end{array}
\right), \quad \text{and} \\ 
M_{2348} & = &  \left(
\begin{array}{cccc}
0&0&0&x_8\\
1&0&0&y_8\\
0&1&0&z_8\\
0&0&1&w_8
\end{array}
\right)
\end{eqnarray*}
is invertible. This implies
\[
\det(M_{1468}) = z_6 y_8 - y_6 z_8 \neq 0,\ \ \det(M_{1278}) = z_7 w_8 - z_8 w_7 \neq 0,\ \ \det(M_{2348}) = x_8 \neq 0.
\]
The condition that $v_8 \in \langle v_5,v_7\rangle$ can be written as
\[
(x_8, y_8, z_8, w_8) = (\beta x_5, \beta \alpha y_6, \beta \alpha z_6, 0) + (0,0,\gamma z_7, \gamma w_7),
\]
for some $\beta, \gamma \in \F_q$.  

We now simplify these algebraic conditions.  Since $x_8 \neq 0$, we see that $x_5 \neq 0$ and $\beta = x_8/x_5$.  If $w_7 = 0$, then $w_8 = 0$, contradicting the condition that $\det(M_{1278}) = z_7 w_8 - z_8 w_7 \neq 0$.  Therefore, $\gamma = w_8/w_7$.  We conclude that 
\begin{eqnarray*}
y_8 &= &  \alpha y_6 (x_8/x_5),\\
z_8 & = & \alpha z_6  (x_8/ x_5) + (w_8/w_7) z_7.
\end{eqnarray*}
We now have 
\[
\det(M_{1278})  = z_7 w_8 - z_8 w_7 = z_7 w_8 - \alpha z_6  (x_8/ x_5) w_7 - w_8 z_7 = -\alpha z_6 w_7 (x_8/x_5) \neq 0,
\]
from which we conclude $\alpha, z_6, w_7, x_8, x_5 \neq 0$.  We also have
\begin{eqnarray*}
\det(M_{1468}) & = &  z_6 y_8 - y_6 z_8 = z_6 \alpha y_6 (x_8/x_5) - y_6 \alpha z_6  (x_8/ x_5) - y_6 (w_8/w_7) z_7\\ 
& = &  -y_6 z_7 (w_8/w_7) \neq 0,
\end{eqnarray*}
from which we conclude that $y_6,z_7, w_8, w_7 \neq 0$.  

Therefore, $\F_q$-vector labelings of $\Gamma$ where the source nodes are labeled with the standard basis vectors are in bijection with the set of nonzero choices for $x_5,y_6, z_6, z_7, w_7, w_8, \alpha$,  which implies that there are $(q-1)^7$ such labelings.  In particular, there exist such labelings over every finite field $\F_q$, and the unique labeling over $\F_2$ is given by the matrix 
\[
\left( 
\begin{array}{cccccccc}
1&0&0&0&1&0&0&1 \\
0&1&0&0&1&1&0&1 \\
0&0&1&0&1&1&1&0 \\
0&0&0&1&0&0&1&1
\end{array}
\right).
\]
\end{example}

Example \ref{vector_label_ex2} shows that even for relatively small code graphs the computations necessary to understand the set of $\F_q$-vector labelings can become intricate.  However, it is straightforward to find the set of $\F_q$-labelings for every particular $q$ using a computer algebra system. 

\subsection{$\F_q$-vector labelings and rational points}

The goal of this section is to understand the set of $\F_q$-vector labelings of a code graph $\Gamma$ in terms of the vanishing and nonvanishing of certain polynomial equations.
\begin{theorem}\label{fq_poly_thm}
Let $\Gamma$ be a code graph with $|\mathcal{S}|$ nodes corresponding to sources and $|\mathcal{Q}|$ nodes corresponding to coding points.  The set of $\F_q$-vector labelings of $\Gamma$ is in bijection with an open subset of a closed affine algebraic subset of $\mathbf{M}_{|\mathcal{S}|\times (|\mathcal{S}|+|\mathcal{Q}|)}(q)$.
\end{theorem}

We break the proof into two lemmas.
\begin{lemma}\label{lem1}
Fix $\mathbf{i} := (i_1, \dots, i_m)$ with $1 \leq i_1 < \dots < i_m \leq n$.  For every $M \in \mathbf{M}_{m \times n}(q)$, write $V_{\mathbf{i}}(M)$ for the set of columns of $M$ indexed by the coordinates of $\mathbf{i}$.   Then $\mathbf{A}_\mathbf{i} := \{M \in \mathbf{M}_{m \times n}(q) \, | \, V_\mathbf{i}(M) \text{ is linearly independent}\}$ is the complement of a hypersurface in $\F_q^{mn}$.
\end{lemma}
\begin{proof}
Note that $\mathbf{M}_{m\times n}(q)$ is identified with the affine space $\F_q^{mn}$. The determinant of the $m \times m$ submatrix with columns given by $V_{\mathbf{i}}(M)$ is a polynomial of degree $m$ in the matrix entries.  Setting this polynomial equal to zero gives a hypersurface in $\F_q^{mn}$.
\end{proof}

\begin{lemma}\label{lem2}
Let $\mathbf{s} := (i_1,\ldots, i_s)$ with $1\le i_1 <\cdots < i_s\le n$ and fix $j \not\in \mathbf{s}$ with $1\le j \le n$.  For every $M \in \mathbf{M}_{m \times n}(q)$, write $V_{\mathbf{s}}(M)$ for the set of columns of $M$ indexed by the coordinates of $\mathbf{s}$ and let $v_j$ denote the $j$\textsuperscript{th} column of $M$.  Then $\mathbf{A}_{\mathbf{s},j} := \{M \in \mathbf{M}_{m \times n}(q) \, | \, v_j \in \langle V_{\mathbf{s}}(M)\rangle\}$ is an open subset of a closed subvariety of $\F_q^{mn}$.  That is, there are two collections of polynomials $\{f_1,\ldots, f_k\}$ and $\{g_1,\ldots, g_l\}$ in the $(s+1)m$ entries of the vectors $v_1,\ldots, v_s, w$ such that $w \in \langle v_1,\ldots, v_s\rangle$ if and only if each $f_i$ simultaneously vanishes and none of the $g_j$ vanish.
\end{lemma}

\begin{proof}
In order to prove this lemma, we need only prove it subject to the additional condition that $\dim\left(\langle V_{\mathbf{s}}(M) \rangle\right) = r$. Taking a union of sets of this form for each $r \in \{0,\ldots, s\}$ completes the proof.

If the matrix with columns $V_{\mathbf{s}}(M)$ has rank $r$, then $v_j \in \langle V_{\mathbf{s}}(M)\rangle\rangle$ if and only if the matrix with columns $V_{\mathbf{s}}(M), v_j$ also has rank $r$.  This condition holds if and only if each of the $(r+1) \times (r+1)$ minors of this $m \times (s+1)$ matrix simultaneously vanishes.  The set of $m \times (s+1)$ matrices of rank $r$ is an open subset of the closed subvariety of $m \times (s+1)$ matrices of rank at most $r$.
\end{proof}
\noindent For more information on these types of \emph{determinantal varieties}, see Lecture 9 of \cite{harris}.

\begin{proof}[Proof of Theorem \ref{fq_poly_thm}]
Suppose $\Gamma$ is a code graph with $|\mathcal{S}|$ nodes corresponding to sources and $|\mathcal{Q}|$ nodes corresponding to coding points.  An $\F_q$-vector labeling of $\Gamma$ gives an element $M \in \mathbf{M}_{|\mathcal{S}| \times (|\mathcal{S}|+|\mathcal{Q}|)}(q)$ such that a finite collection of subsets of $|\mathcal{S}|$ columns are linearly independent, and there are finitely many collections of subsets of columns $V_{\mathbf{s}}(M), v_j$ such that $v_j \in \langle V_{\mathbf{s}}(M)\rangle$.  Combining Lemmas \ref{lem1} and \ref{lem2} completes the proof.
\end{proof}

\subsection{$\F_q$-vector labelings and Grassmannians}\label{sec_grass}

Suppose $\Gamma$ is a code graph with $|\mathcal{S}|$ vertices corresponding to sources and $|\mathcal{Q}|$ vertices corresponding to coding points.  Let $M \in \mathbf{M}_{|\mathcal{S}| \times (|\mathcal{S}|+|\mathcal{Q}|)}(q)$ be the matrix corresponding to an $\F_q$-vector labeling of $\Gamma$.  Scaling every column of $M$ gives another $\F_q$-vector labeling, so in order to completely understand the set of $\F_q$-labelings we need only understand this set up to scalings of columns.  

Let $C \subseteq \F_q^{|\mathcal{S}|+|\mathcal{Q}|}$ be the subspace spanned by the rows of $M$.  Choosing a different basis for this subspace gives another $\F_q$-labeling of $\Gamma$.  This motivates studying $\F_q$-labelings in terms of subsets of a Grassmannian.

\begin{definition}
Let $V$ be a vector space over a field.  For positive integers $k\le n$, let $G(k,n)$ be the set of $k$-dimensional linear subspaces of $V$.  This set has the structure of an algebraic variety and is called the \emph{Grassmannian}.
\end{definition}
\noindent For a brief introduction to the Grassmannian, see Lecture 6 of \cite{harris}. We recall some notation from \cite{skorobogatov}.  

\begin{definition}
Let $G(k,n)$ be the Grassmannian of $k$-dimensional subspaces of $\F_q^n$.  Let $e_i$ be the vector whose $i$\textsuperscript{th} coordinate is $1$, and the other coordinates are zero.  For $I \subseteq \{1,\ldots, n\}$ let $W_I = \langle \{e_i\}_{i \in I}\rangle$. Let $f$ be a function from the subsets of $\{1,\ldots, n\}$ to $\Z_{\ge 0}$.  

Let $U_f(k,n)$ be the set of vector subspaces $C \subset \F_q^n$ such that for all $I \subset \{1,\ldots, n\},\ \dim(C \cap W_I) = f(I)$.  There is a stratification of $G(k,n)$ given by
\[
G(k,n) = \bigcup_f U_f(k,n).
\]

\end{definition}
\noindent See \cite{skorobogatov} for some history of the study of this stratification.

An $\F_q$-vector labeling of $\Gamma$ gives an element of $M \in \mathbf{M}_{|\mathcal{S}| \times (|\mathcal{S}|+|\mathcal{Q}|)}(q)$. Since the columns corresponding to the source nodes are linearly independent, the row span of $M$ is $|\mathcal{S}|$-dimensional, and $M$ gives an element of $G\left(|\mathcal{S}|, |\mathcal{S}|+|\mathcal{Q}|\right)$.
\begin{proposition}\label{strata_prop}
Let $\Gamma$ be a code graph with $|\mathcal{S}|$ nodes corresponding to sources and $|\mathcal{Q}|$ nodes corresponding to coding points.  The $\F_q$-vector labelings of $\Gamma$  are in bijection, up to a choice of basis, with the $\F_q$-points of 
\[
\bigcup_f U_f\left(|\mathcal{S}|, |\mathcal{S}|+ |\mathcal{Q}|\right)
\]
for some union of functions $f$ from the subsets of $\left\{1,\ldots, |\mathcal{S}|+|\mathcal{Q}|\right\}$ to $\Z_{\ge 0}$ determined by $\Gamma$.
\end{proposition}

\begin{proof}
Let $M \in \mathbf{M}_{|\mathcal{S}| \times (|\mathcal{S}|+|\mathcal{Q}|)}(q)$ be the matrix corresponding to an $\F_q$-labeling of $\Gamma$. Let $v_i$ denote the $i$\textsuperscript{th} column of $M$.  Let $C \subseteq \F_q^{|\mathcal{S}|+|\mathcal{Q}|}$ be the subspace generated by the rows of $M$.

There are two types of conditions required for $M$ to correspond to an $\F_q$-vector labeling of $\Gamma$: 
\begin{enumerate}
\item Subsets $\mathbf{i} := (i_1, \dots, i_m)$ with $1 \leq i_1 < \dots < i_m \leq n$ such that the set of columns $V_{\mathbf{i}}(M)$ indexed by the coordinates of $\mathbf{i}$ gives an invertible $|\mathcal{S}|\times |\mathcal{S}|$ matrix.

\item Subsets $\mathbf{s} := (i_1,\ldots, i_s)$ with $1\le i_1 <\cdots < i_s\le n$ and $j \not\in \mathbf{s}$ with $1\le j \le n$ such that $v_j \in \langle V_{\mathbf{s}}(M)\rangle$, where $V_{\mathbf{s}}(M)$ denotes the set of columns of $M$ indexed by the coordinates of $\mathbf{s}$

\end{enumerate}

The first condition is equivalent to $\dim(C \cap W_{\mathbf{i}}) = |\mathcal{S}|$.  The second condition says that if $\dim(C \cap W_{\mathbf{s}}) = r$ then $\dim(C \cap W_{\mathbf{s} \cup j}) = r$.  Each condition can be expressed in terms of the function $f$ for which $C \in U_f(|\mathcal{S}|, |\mathcal{S}|+|\mathcal{Q}|)$.
\end{proof}

Once $n$ is large, analyzing how the number of $\F_q$-points of $U_f(k,n)$ varies as a function of $q$ becomes intricate.  Let $p$ be a prime. Given $k \ge 3$ and any closed projective algebraic set over $\F_p$, a version of Mn\"ev's Universality Theorem says that there exists $n$ and $f$ such that $X$ is isomorphic to the Zariski closure of $U_f(k,n)$ under the action of the diagonal matrices \cite{skorobogatov}.  This implies that for large values of $k$ and $n$ the function counting $\F_q$-points of $U_f(k,n)$ can become difficult to understand.  In Examples \ref{butterfly_label} and \ref{vector_label_ex2} we saw code graphs for which the number of $\F_q$-labelings was given by a polynomial in $q$.  We give one more example to show that this does not always occur and that the number of $\F_q$-vector labelings of $\Gamma$ does not have to increase monotonically with $q$.

Note that the multicast network in the example is not reduced and the presented code graph does not have the minimal number of coding points. Nonetheless, this example demonstrates interesting algebraic properties and has already occurred in the network coding literature. 
\begin{example}\cite[Sun, Yin, Li, Long]{su15}

\begin{figure}[htb]
\centering
\begin{subfigure}{.48\textwidth}
\centering
\scalebox{.29}{\includegraphics{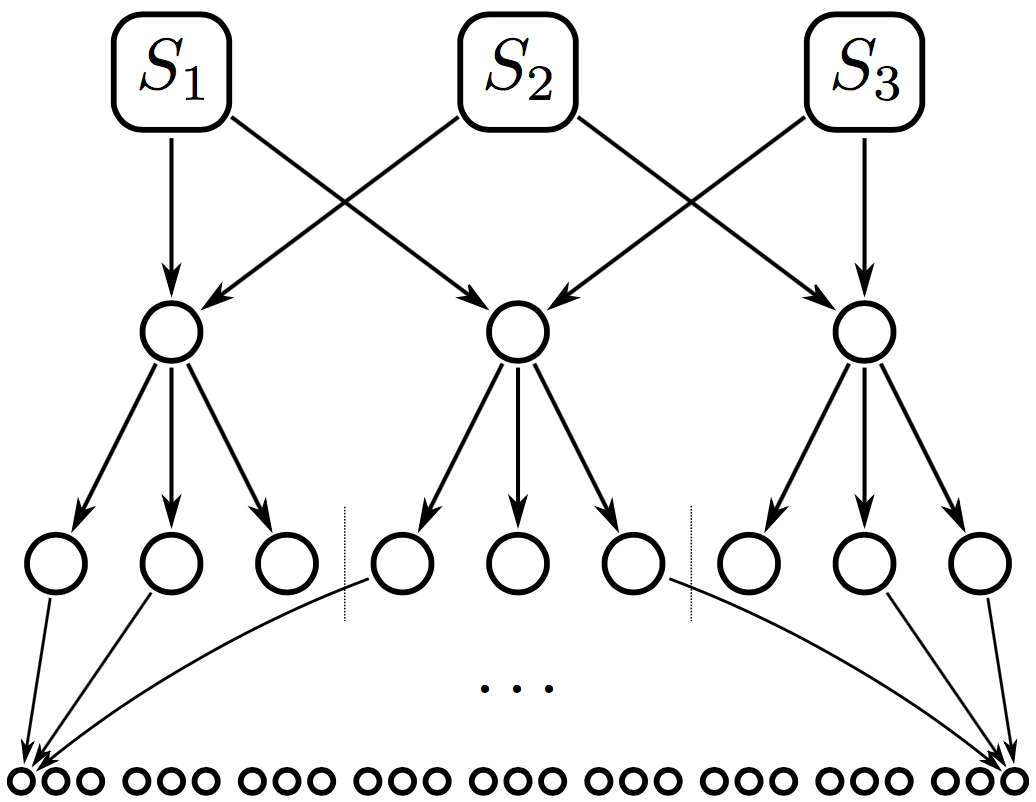}}
\caption{Combination network.}
\label{f:combination}
\end{subfigure}\hfill
\begin{subfigure}{.48\textwidth}
\centering
\scalebox{.29}{\includegraphics{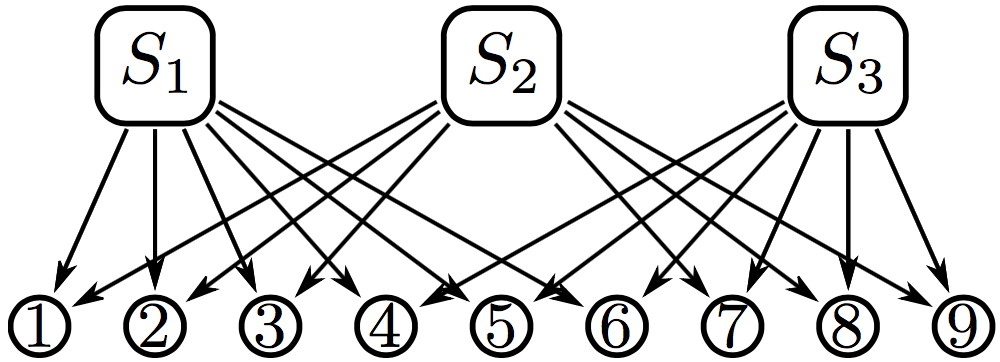}}
\caption{Code Graph.}
\label{f:combination_codegraph}
\end{subfigure}
\caption{Combination network and its code graph.}
\end{figure}

We consider a particular example of a class of networks called \emph{combination networks}.  These networks are described in general in \cite{su15}, and we consider the example for $m = 1$.  This is a multicast network with $3$ sources, $9$ coding points, and $81$ receivers. This leads to the code graph $\Gamma$ of Figure \ref{f:combination_codegraph} where the source nodes are labeled with the empty set and each coding node $i\in \{1,\dots, 9\}$ is labeled with a set $L_i$ of $27$ receivers such that:
\begin{itemize}
\item $|\mathbf{V}_R|=3$ for any $R\in \mathcal{R}$, where $\mathbf{V}_R$ is the set of coding nodes labeled with receiver $R$; 
\item $L_1\cap L_2 \cap L_3=L_4\cap L_5 \cap L_6=L_7\cap L_8 \cap L_9=\emptyset$.
\end{itemize} 

Choosing the standard basis vectors for the labels of the three source nodes, and scaling each column so that its first nonzero entry is equal to $1$, leads to an element of $\mathbf{M}_{3\times 12}(q)$ of the form
\[
\left( 
\begin{array}{cccccccccccc}
1&0&0&1&1&1&0&0&0&1&1&1 \\
0&1&0&\alpha_1&\alpha_2&\alpha_3&1&1&1& 0&0&0 \\
0&0&1&0&0&0&\beta_1&\beta_2&\beta_3&\gamma_1&\gamma_2&\gamma_3 
\end{array}
\right),
\]
subject to the constraints that each $\alpha_i, \beta_i, \gamma_i \in \F_q^*,\ 1\le i \le 3$ and
\[
\alpha_i \neq \alpha_j,\ \   \beta_i \neq \beta_j,\ \   \gamma_i \neq \gamma_j,\  \ \forall 1\le i<j\le 3
\]
\[
\{\gamma_1, \gamma_2, \gamma_3 \}  \subseteq  \F_q^* \setminus \{-\alpha_i \beta_j \ \mid  1\le i,j \le 3\}.
\]

The $\F_q$-vector labelings of $\Gamma$ are therefore in bijection with the complement of a (reducible) hypersurface in a $9$-dimensional affine space over $\F_q$.  Using a computer algebra system such as Sage or Magma, one can compute the number of $\F_q$-vector labelings over small finite fields.  If there were a polynomial formula for the number of such labelings it would have degree at most $9$.  Computing the number of $\F_q$-vector labelings for each $q\le 40$ shows that no such polynomial formula exists.

\end{example}

\subsection{An open question on $\F_q$-vector labelings of code graphs}

Let $G$ be a multicast network with $N$ receivers and code graph $\Gamma$.  In the introduction we mentioned the result of Ho \emph{et al.}, which implies that the capacity of $G$ is achievable over all finite fields $\F_q$ with $q > N$ \cite{ho06}. The proof of this theorem shows that there is an $\F_q$-vector labeling of $\Gamma$ for all $q> N$. In particular, this means that for every multicast network and every prime $p$ there exists a valid $\F_{p^k}$-vector labeling for all sufficiently large $k$.

For a matrix $M \in \mathbf{M}_{m \times n}(q)$, let $v_i$ denote the $i$\textsuperscript{th} column of $M$.
\begin{question}
Consider a finite collection $A$ of $m$-element subsets of $\{1,\ldots, n\}$ and a finite collection $B$ of pairs $(\mathbf{s},k)$ such that $\mathbf{s}$ is a subset of $\{1,\ldots, n\}$ not including $k$.  Proposition \ref{strata_prop} shows that the set of matrices $M \in \mathbf{M}_{m \times n}(q)$, up to a choice of basis, such that each element of $A$ gives a linearly independent set of columns, and each subset $B$ gives a set of columns where $v_k$ is in the span of the columns corresponding to $\mathbf{s}$, leads in a natural way to finite unions of sets of the form $U_f(m,n)$.

Which unions of this form can arise from the set of $\F_q$-labelings of the code graph of a multicast network?
\end{question}

We argue that not every union of sets of the form $U_f(k,n)$ can arise from a multicast network.  We describe a finite collection of algebraic conditions on the columns of a $3 \times 7$ matrix such that the number of matrices $M \in \mathbf{M}_{3\times 7}(q)$ satisfying these conditions is zero for all odd $q$, but is nonzero for all even $q$.  These conditions come from the \emph{Fano plane}, a special configuration of $7$ points and $7$ lines, with each line containing exactly $3$ of the points.  Figure \ref{Fano} represents this incidence structure of the points and lines of the Fano plane. The circle connecting points $\{1,2,3\}$ counts as a line.
\begin{figure}[h]
\centering
\scalebox{.27}{\includegraphics{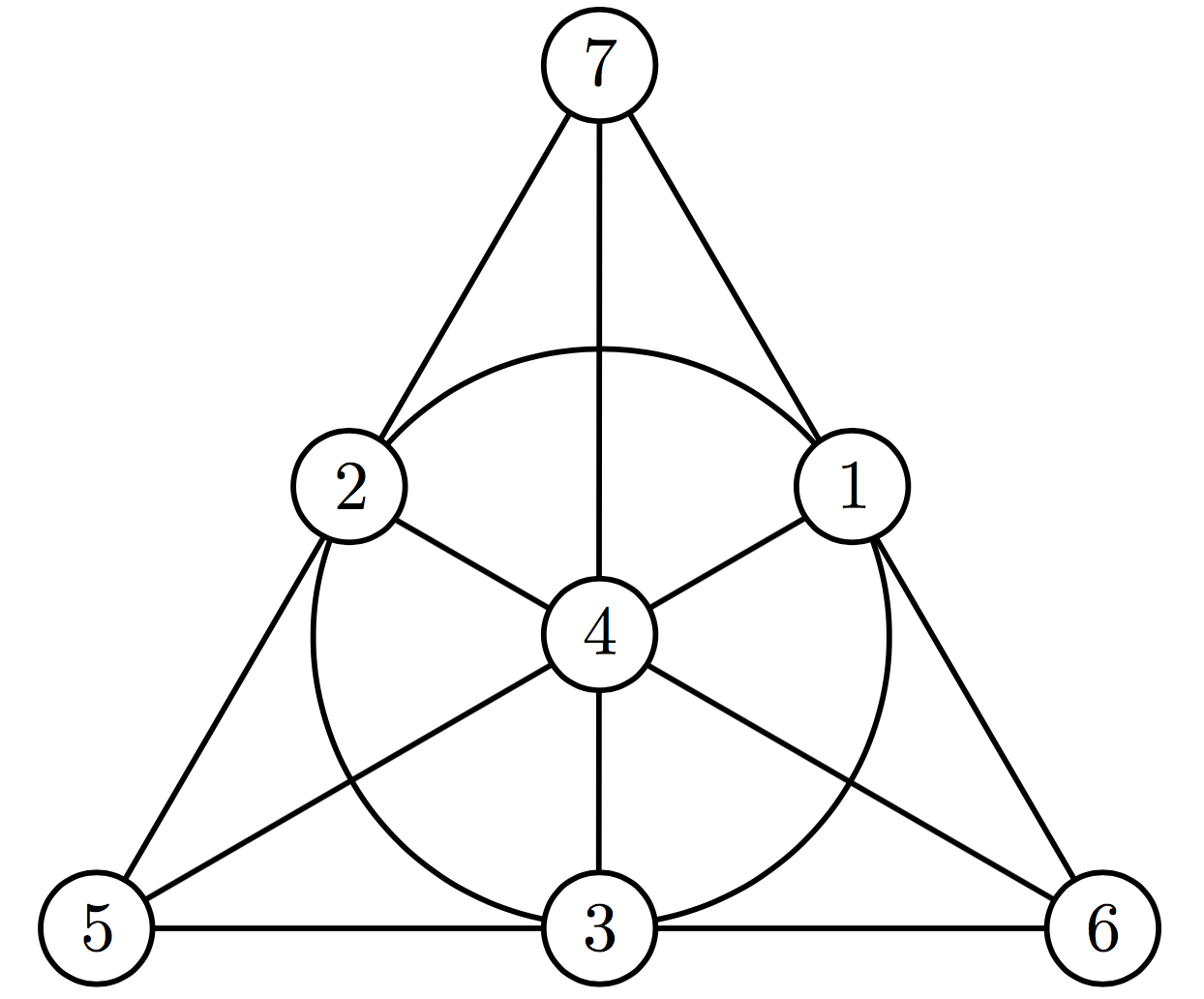}}
\caption{The Fano plane, a special configuration of $7$ points and $7$ lines.}
\label{Fano}
\end{figure}

Let $M \in \mathbf{M}_{3\times 7}(q)$ have columns $v_1, \ldots, v_7$. We consider the set of $M$ such that the following conditions hold:
\begin{eqnarray*}
v_3 & \in &  \langle v_5, v_6 \rangle,\ v_2 \in \langle v_5, v_7\rangle,\ v_1 \in \langle v_6, v_7\rangle,\\ 
v_3 & \in &  \langle v_7, v_4\rangle,\ v_1 \in \langle v_4,v_5 \rangle,\ v_2 \in \langle v_4, v_6\rangle,\ v_3 \in \langle v_1, v_2\rangle,
\end{eqnarray*}
and all other collections of $3$ vectors $v_i, v_j, v_k$ not occuring in this list are linearly independent; in other words, the linear dependencies are determined by the lines of the Fano plane.

Choosing an appropriate basis for the $3$-dimensional subspace of $\F_q^7$ spanned by the rows of $M$ we may assume that 
\[
v_5 = (1,0,0),\ v_6 = (0,1,0),\ v_7 = (0,0,1),\ v_4 = (1,1,1).
\]
Since $v_3 \in \langle v_5,v_6\rangle \cap \langle v_4, v_7\rangle$, multiplying by an appropriate scalar we may further assume that $v_3 = (1,1,0)$.  Similarly, we may assume that $v_1 = (0,1,1)$ and $v_2 = (1,0,1)$.  The condition that $v_3 \in \langle v_1, v_2\rangle$ implies that, in $\F_q$, 
\[
0 = \det \left( 
\begin{array}{ccc}
1 & 1 & 0 \\
1 & 0 & 1 \\
0 & 1 & 1 
\end{array}
\right) = -2.
\]
Such a matrix can only arise in characteristic $2$, a result that is discussed in detail in the projective geometry literature; see, e.g., \cite{grunbaum}.  Since the matrix problem specified by the Fano plane does not have solutions over all sufficiently large finite fields $\F_q$, it cannot `come from' the $\F_q$-vector labelings of a code graph.

\newcommand{\etalchar}[1]{$^{#1}$}

\end{document}